\documentclass[11pt]{article}
\usepackage[english]{babel}
\usepackage{dsfont}
\usepackage{graphicx}
\usepackage{setspace}
\usepackage[margin=1in]{geometry}
\usepackage{amsmath}
\usepackage{amssymb}
\usepackage{amsthm}
\usepackage{natbib}
\usepackage{enumerate}
\usepackage{rotating}
\usepackage{pdflscape}
\usepackage{makecell}
\usepackage{multicol}
\usepackage{multirow}
\usepackage{tabularx}
\usepackage[title]{appendix}

\usepackage{etoolbox} 
\usepackage{color}
\usepackage{xcolor}
\usepackage{colortbl}
\usepackage[colorinlistoftodos,prependcaption,textsize=tiny]{todonotes}
\usepackage[colorlinks,citecolor=blue]{hyperref}
\hypersetup{pdfauthor={Mogstad, Romano, Shaikh, and Wilhelm}}

\usepackage[nameinlink,noabbrev]{cleveref}
\usepackage{caption, floatrow, setspace}
\usepackage{hhline}
\usepackage{tabularx}
\usepackage[bottom]{footmisc}
\usepackage[final]{pdfpages}
\usepackage{booktabs}
\usepackage{bm}
\usepackage{esvect}
\newcolumntype{Y}{>{\centering\arraybackslash}X}
\DeclareFloatVCode{myrowsep}{\vskip 4ex}

\def\qed{\rule{2mm}{2mm}}

\usepackage[pagewise,mathlines]{lineno}
\mathchardef\dash="2D

\newtheorem{theorem}{Theorem}[section]
\newtheorem{lemma}{Lemma}[section]

\newtheorem{corollary}{Corollary}[section]

\theoremstyle{definition}

\newtheorem{remark}{Remark}[section]

\newtheorem{algorithm}{Algorithm}[section]

\AtEndEnvironment{remark}{~\qed}
\AtEndEnvironment{example}{~\qed}

\begin{document}
\author{
Sergei Bazylik \\
Department of Economics\\
University of Chicago 
\and
Magne Mogstad \\
Department of Economics\\
University of Chicago \\
Statistics Norway \& NBER 
\and
Joseph P.\ Romano \\
Departments of Statistics and Economics\\
Stanford University 
\and
Azeem M.\ Shaikh\\
Department of Economics\\
University of Chicago 
\and
Daniel Wilhelm\\
Departments of Statistics and Economics\\
LMU Munich 
}

\bigskip

\title{\vspace{-1in}Finite- and Large-Sample Inference for Ranks using \\ Multinomial Data with an Application to Ranking Political Parties \thanks{The third author acknowledges support from the National Science Foundation (MMS-1949845). The fourth author acknowledges support from the National Science Foundation (SES-1530661). The fifth author acknowledges support from the ESRC Centre for Microdata Methods and Practice at IFS (RES-589-28-0001) and the European Research Council (Starting Grant No. 852332).
}}
\maketitle

\vspace{-0.3in}

\begin{spacing}{1}
\begin{abstract}
\noindent It is common to rank different categories by means of preferences that are revealed through  data on choices.  A prominent example is the ranking of political candidates or parties using the estimated share of support each one receives in surveys or polls about political attitudes.  Since these rankings are computed using estimates of the share of support rather than the true share of support, there may be considerable uncertainty concerning the true ranking of the political candidates or parties.   In this paper, we consider the problem of accounting for such uncertainty by constructing confidence sets for the rank of each category. We consider both the problem of constructing marginal confidence sets for the rank of a particular category as well as simultaneous confidence sets for the ranks of all categories. A distinguishing feature of our analysis is that we exploit the multinomial structure of the data to develop confidence sets that are valid in finite samples.  We additionally develop confidence sets using the bootstrap that are valid only approximately in large samples.  We use our methodology to rank political parties in Australia using data from the 2019 Australian Election Survey.  We find that our finite-sample confidence sets are informative across the entire ranking of political parties, even in Australian territories with few survey respondents and/or with parties that are chosen by only a small share of the survey respondents.  In contrast, the bootstrap-based confidence sets may sometimes be considerably less informative. These findings motivate us to compare these methods in an empirically-driven simulation study, in which we conclude that our finite-sample confidence sets often perform better than their large-sample, bootstrap-based counterparts, especially in settings that resemble our empirical application.
\end{abstract}
\end{spacing}



\thispagestyle{empty} 
\newpage

\section{Introduction} \label{sec:intro}

Preferences over different categories are often assessed by means of data on choices. It is natural to summarize this type of data by ranking the different categories according to the share of support each one receives in the data. A prominent example is provided by surveys or polls of support for political candidates. In this case, individuals choose one political candidate or party from among those available. The resulting rankings may shape public discussion, inform campaigns, and be used as inputs into consequential decisions before the actual election. For example, in U.S. Presidential Elections, the decision about which candidates to feature in nationally televised political debates may hinge on their performance in different polls leading up to the election. Another prominent example is choice-based conjoint analysis, in which respondents select which of several options they would purchase or otherwise choose if given the option.  Such analyses are regularly used in marketing analysis both to assess which product features are most valued and thereby inform decisions about which products to introduce. 
In both these examples, however, it is important to acknowledge the uncertainty surrounding these rankings.

Such data on choices, including polls of political attitudes, commonly feature limited sample sizes and/or categories whose true share of support is small.  As explained further below, these features pose challenges to inference methods justified using large-sample arguments. In contrast, this paper considers the problem of constructing confidence sets for the rank of each category that are valid in finite samples, even when some categories are chosen with probability close to zero. We consider two types of confidence sets: marginal confidence sets for the rank of a particular category, by which we mean random sets that contain the rank of a particular category with probability no less than some pre-specified level, as well as simultaneous confidence sets for the ranks of all categories, by which we mean random sets that contain the ranks of all categories with probability  no less than some pre-specified level. The former confidence sets provide a way of accounting for uncertainty when answering questions pertaining to the rank of a particular category, whereas the latter confidence sets provide a way of accounting for uncertainty when answering questions pertaining to the ranks of all categories.  Our constructions are based off of testing a family of one-sided null hypotheses concerning differences in pairs of success probabilities in a way that controls the familywise error rate in finite samples. In order to do so, we exploit the multinomial structure of the data, which enables the use of a simple conditioning argument.

As a second contribution, we develop bootstrap methods for the construction of these confidence sets.  
Their validity is justified using large-sample arguments as in \cite{Mogstad:2024aa}. However, unlike in \cite{Mogstad:2024aa}'s applications, the estimators of the success probabilities of different categories are necessarily dependent and the bootstrap procedure proposed in this paper explicitly accounts for this dependence. As described in more detail below, the results in \cite{Brown:2001p5787} suggest that such bootstrap methods may perform poorly when sample sizes are small and/or some categories are chosen with small probabilities. In particular, approaches that explicitly or implicitly (such as the bootstrap) rely on asymptotic normality of the estimators of the success probabilities perform poorly when the true success probability is small. In such a case, it is well known that the Poisson distribution is in fact a better approximation than the normal. In our simulations, we find not only that the bootstrap-based confidence sets (with or without studentization) may have coverage probability considerably below the desired level, but also that they may be excessively wide. In contrast, the finite-sample confidence sets have coverage probability no less than the desired nominal level and may even be considerably shorter.

We apply our inference procedures to re-examine the ranking of political parties in Australia using data from the 2019 Australian Election Survey.  We find that the finite-sample (marginal and simultaneous) confidence sets are remarkably informative across the entire ranking of political parties, even in Australian territories with few survey respondents and/or with parties that are chosen by only a small share of the survey respondents.  To illustrate this point further, consider one particular Australian territory, Greater Melbourne.  We find that the finite-sample confidence sets are either of similar length to or substantially shorter than their bootstrap-based counterparts (with or without studentization).  For instance, at conventional significance levels, the finite-sample marginal confidence set for the rank of the Green Party contains only rank 4.  In contrast, the bootstrap-based marginal confidence sets (with or without studentization) contain the ranks 3 to 7, thus exhibiting substantially more uncertainty about the true rank of the Green Party.  The studentized procedure leads to especially wide confidence sets for the ranks of parties that are chosen only by a small share of respondents.  We find similar patterns in the eight most populous territories, while confidence sets in the remaining seven least populous territories are uninformative due to very small sample sizes. Unlike in Greater Melbourne, however, in some other territories bootstrap-based confidence sets (with or without studentization) may be slightly smaller than their finite-sample counterparts.

These findings motivate us to compare the different confidence sets in a simulation study modeled after our empirical application.  The findings of this exercise can be summarized as follows. First, finite-sample marginal confidence sets have coverage probabilities no less than the desired nominal level in all simulation designs, including those with very small sample sizes and/or with parties that are chosen by only a small share of the respondents.  Second, bootstrap-based confidence sets without studentiziation also have coverage probabilities no less than the nominal level, except when sample sizes are very small. In contrast, bootstrap-based confidence sets with studentization may have coverage probabilities less than the nominal level when sample sizes are small and/or the number of parties to be ranked is not too small. Third, the finite-sample confidence sets may produce considerably shorter confidence sets for the ranks compared to the bootstrap-based ones, especially when there are parties that are chosen by only a small share of respondents. However, there are also situations in which the latter methods produce shorter confidence sets than the former, so neither approach always dominates the other in terms of length of their confidence sets.

Our paper is most closely related to the aforementioned paper by \cite{Mogstad:2024aa}.  We emphasize that the primary contribution of our analysis relative to theirs is to show how one may exploit additional structure given by the multinomial data to construct confidence sets that enjoy finite-sample validity. Importantly, the finite sample guarantee allows for data-generating processes with success probabilities that are arbitrarily close to zero, whereas the asymptotic arguments justifying the bootstrap require these probabilities to be bounded away from zero. Second, we propose a bootstrap method that accounts for the dependence in the estimators of the multinomial success probabilities. Our paper is also related to a recent paper by \cite{Klein:2020oi}, who consider the problem of constructing confidence sets analogous to those in \cite{Mogstad:2024aa}.  We show how a modification of their procedure can also be used to construct confidence sets that are valid in finite samples in the presence of multinomial data.  In our simulations, we find that the resulting confidence sets are often of comparable length to our finite-sample confidence sets, but sometimes meaningfully larger.  We refer the reader to \cite{Mogstad:2024aa} for additional comparisons.  Other related work includes \cite{Goldstein:1996re}, who propose a different bootstrap-based confidence set to account for uncertainty in reported ranks.  As explained by \cite{Hall:2009oi}, \cite{Xie:2009oi} and most recently by \cite{Mogstad:2024aa}, however, this method performs poorly in the presence of categories that are chosen with similar frequencies (i.e., in the context of our simulations, when some parties are nearly tied).  We confirm this finding in our simulations.  

Our paper also draws motivation from earlier work by \cite{Brown:2001p5787}, who demonstrate that conventional confidence intervals for the probability of success using binomial data may behave poorly in the sense of exhibiting undercoverage, especially when the success probability is close to zero or one, and may also behave erratically in the sense that coverage probabilities may be volatile and non-monotonic in the sample size. In our simulations, we find similar patterns concerning the coverage probabilities for the differences in pairs of success probabilities using multinomial data.  For this reason, we view insistence upon finite-sample validity for our confidence sets to be especially compelling in this setting.  On the other hand, we find that this poor behavior of confidence sets for the differences in the success probabilities need not translate into similar behavior for the implied confidence sets for the ranks.

We note some key differences between the problems considered in this paper and those of two recent papers in econometrics, \cite{Andrews:2018qf} and \cite{gu2020invidious}.  In the context of the multinomial setting studied here, \cite{Andrews:2018qf} develop methods for inference on the true success probability for the randomly selected category whose estimated rank is highest.  In contrast, as the discussion above makes clear, we develop methods for inference on the true ranks themselves.  \cite{gu2020invidious} develop decision rules for selecting the most popular categories (i.e., those with the highest success probabilities), which is more closely related to a literature on subset selection (see \cite{Gupta:1979hi} for a review).  We show, however, how our simultaneous confidence sets may be used to create a complimentary object that we refer to as the confidence set for the $\tau$-best.  For given value of $\tau$, such a confidence set is a random set that contains the identities of (all of) the categories whose rank is less than or equal to $\tau$ with probability approximately no less than some pre-specified level. 

Finally, there is also a large literature about the analysis of datasets that contain observations of elicited rankings (e.g., \cite{Marden:1995aa}), but this differs from our setting in which we assume only an individual's top choice is observed rather than their ranking of all available options.

The remainder of the paper is organized as follows.  In Section \ref{sec:setup}, we introduce our general setup, including a formal description of the different types of confidence sets we consider and the general testing problem involved in their constructions.  Suitable tests that lead to confidence sets that are valid in finite samples are then described in \ref{sec:margsimul}.  The construction of confidence sets for the $\tau$-best that are valid in finite samples is briefly summarized in \ref{sec:taubest}.  Section \ref{sec: bootstrap} described bootstrap-based versions of these same confidence sets.  In Section \ref{sec: Melbourne}, we apply our inference procedures to re-examine the ranking of political parties in Australia.  Finally, in Section \ref{sec: sim}, we examine the finite-sample behavior of our inference procedures via a simulation study modeled after our empirical application. 

\section{Main Results} 
\label{sec:main}

\subsection{Setup and Notation}
\label{sec:setup}

Let $j \in J \equiv \{1, \ldots, p\}$ index categories of interest, e.g., parties in an election. 
There are $n$ independent observations, and each observation falls in category $j$ with probability $\theta_j$.
Let
$X_j$ denote the observed count for category $j$ from  the $n$ observations, e.g. the number of votes party $j$ receives from $n$ voters.   Hence, $X\equiv (X_1,\ldots,X_p)'$ is distributed according to the multinomial distribution with parameters $n$ and $\theta\equiv (\theta_1,\ldots,\theta_p)'$. 

The rank of category $j$ is defined as
$$r_j \equiv 1 + \sum_{k \in J} \mathds{1}\{\theta_k > \theta_j \},$$
where $\mathds{1}\{A\}$ is equal to one if the event $A$ holds and equal to zero otherwise. Let $r\equiv (r_1,\ldots,r_p)'$. Before proceeding, it is useful to provide a simple example to illustrate the way in which ties are handled with this definition of ranks: if $\theta = (0.4,0.1,0.1,0.2,0.2)'$, then $r = (1,4,4,2,2)'$.

The primary goal is to construct confidence sets for the rank of a particular category or for the ranks of multiple categories simultaneously. Let $J_0\subseteq J$ denote the categories of interest. For a given value of $\alpha \in (0,1)$, we use data $X$ to construct (random) sets $$R_{n} \equiv \prod_{j\in J_0} R_{n,j},$$ where the (random) sets $R_{n,j}$, $j\in J_0$, are such that
\begin{equation} \label{eq: coverage}
 P\left\{r_{j} \in R_{n,j}\;\forall j\in J_0\right\} \geq 1 - \alpha.
\end{equation}
If $J_0$ is a singleton, then sets $R_{n}$ satisfying \eqref{eq: coverage} are referred to as {\it marginal confidence sets for the rank of a single category}. If $J_0=J$, then sets $R_{n}$ satisfying \eqref{eq: coverage} are referred to as {\it simultaneous confidence sets for the ranks of all categories}. The remainder of the paper, however, allows $J_0$ to be any subset of $J$. In our constructions, $R_{n,j}$ are subsets of $J$ for each $j\in J_0$, allowing for the possibility that the lower endpoint is 1 or the upper endpoint is $p$ to permit both one-sided and two-sided inference.

In addition, we consider the goal of constructing confidence sets for the identities of all categories whose rank is less than or equal to a pre-specified value $\tau \in J$, i.e., for a given value of $\alpha \in (0,1)$, we construct (random) sets $R^{\tau-\rm{best}}_n$ that are subsets of $J$ and satisfy 
\begin{equation} \label{eq:taubestcoverage}
	P\left\{R_0^{\tau-\rm{best}} \subseteq R^{\tau-\rm{best}}_n\right\} \geq 1 - \alpha~,
\end{equation}
where $$R_0^{\tau-\rm{best}} \equiv \{j \in J : r_j \leq \tau \}~.$$  Sets satisfying \eqref{eq:taubestcoverage} are referred to as {\it confidence sets for the $\tau$-best categories}.

As in \cite{Mogstad:2024aa}, the construction of confidence sets for ranks can be based on tests of the hypotheses
$$H_{j,k}\colon \theta_j\leq \theta_k $$
for pairs of indices $(j,k)\in J^2$. Which pairs are relevant depends on whether the desired confidence sets for the ranks indicated by $J_0$ are lower, upper or two-sided confidence bounds:
\begin{align*}
	J^{\rm lower} &\equiv \{(j,k)\in J\times J_0\colon j\neq k \}\\
	J^{\rm upper} &\equiv \{(j,k)\in J_0\times J\colon j\neq k \}\\
	J^{\rm two-sided} &\equiv J^{\rm lower} \cup J^{\rm upper}
\end{align*}
Suppose a family of tests of the hypotheses $H_{j,k}$ is given.  Then, for each $j\in J_0$, let
\begin{equation}\label{eq: rej minus}
	\text{Rej}_{j}^{-} \equiv \{k\in J\setminus \{j\}\colon \text{ reject } H_{k,{j}} \text{ and claim } \theta_{j}<\theta_k\}
\end{equation}
indicate the set of hypotheses that are rejected in favor of $\theta_{j}<\theta_k$ and 
\begin{equation}\label{eq: rej plus}
	\text{Rej}_{j}^{+} \equiv \{k\in J\setminus\{j\} \colon \text{ reject } H_{j,k} \text{ and claim } \theta_{j}>\theta_k\}
\end{equation}
the set of hypotheses that are rejected in favor of $\theta_{j}>\theta_k$. Consider the goal of constructing a two-sided marginal confidence set for the rank of a category $j_0$, i.e., $J_0=\{j_0\}$. Then, $\text{Rej}_{j_0}^{-}$ contains all categories $k\neq j_0$ whose parameter $\theta_k$ is claimed to be strictly larger than $\theta_{j_0}$. If these claims were correct, then the lower bound on the rank of category $j_0$ must be equal to the number of such categories $k$, denoted by $|\text{Rej}_{j_0}^{-}|$, plus one. Similarly, $\text{Rej}_{j_0}^{+}$ contains all categories $k\neq j_0$ whose parameter $\theta_k$ is claimed to be strictly smaller than $\theta_{j_0}$. Again, if these claims were correct, then the upper bound on the rank of category $j_0$ must be the total number of categories, $p$, minus the number of categories with smaller probability, denoted by $|\text{Rej}_{j_0}^{+}|$. Therefore, if all claims made in $\text{Rej}_{j_0}^{-}$ and $\text{Rej}_{j_0}^{+}$ are correct, the set
\begin{equation*}
	R_{n,j_0} \equiv \left\{|\text{Rej}_{j_0}^-| + 1, \ldots, p-|\text{Rej}_{j_0}^+|\right\}.
\end{equation*}
contains the rank of category $j_0$, $r_{j_0}$. More generally, for an arbitrary set of indices $J_0\subseteq J$, if all claims made in $\text{Rej}_{j}^{-}$ and $\text{Rej}_{j}^{+}$, $j\in J_0$, are correct, then the set 
\begin{equation}\label{eq: def joint CS}
	R_{n} \equiv \prod_{j\in J_0} R_{n,j}\qquad\text{ with }\qquad R_{n,j}\equiv\left\{|\text{Rej}_j^-| + 1, \ldots, p-|\text{Rej}_j^+|\right\}
\end{equation}
contains the ranks of all categories in $J_0$, $(r_j\colon j\in J_0)$. Of course, it cannot be guaranteed that tests of $H_{j,k}$ never falsely reject, but the probability of such mistakes can be controlled. More specifically, for the set $R_{n}$ to satisfy the coverage statement in \eqref{eq: coverage}, the number of false claims must be controlled in the sense that the familywise error rate for testing $H_{j,k}$ for the relevant pairs of indices $I\subset J^2$ is no larger than $\alpha$, i.e.,
\begin{equation}\label{eq: FWER control}
	 FWER_I\equiv P\left\{ \text{reject at least one true hypothesis } H_{j,k},\; (j,k)\in I \right\}\leq \alpha.
\end{equation}
For two-sided confidence sets, the relevant set of indices $I$ is $J^{\rm two-sided}$ and, for one-sided confidence sets, it is either $J^{\rm lower}$ or $J^{\rm upper}$. The following theorem is a slight generalization (allowing for a general set $J_0$ of indices) of Theorem~3.4 in \cite{Mogstad:2024aa} and summarizes the above discussion.

\begin{theorem}\label{thm: coverage}
	For $J_0\subseteq J$, let $I$ be equal to $J^{\rm lower}$, $J^{\rm upper}$, or $J^{\rm two-sided}$. Let $R_{n}$ be defined by \eqref{eq: rej minus}, \eqref{eq: rej plus}, and \eqref{eq: def joint CS}, where the family of hypotheses $H_{j,k}$, $(j,k)\in I$, is tested using a procedure that satisfies \eqref{eq: FWER control} for some $\alpha\in(0,1)$. Then, $R_{n}$ satisfies \eqref{eq: coverage}.
\end{theorem}

Instead of controlling the coverage probability in finite samples as in \eqref{eq: coverage}, the confidence sets proposed in \cite{Mogstad:2024aa} only asymptotically control the coverage probability. Their constructions assume the availability of bootstrap confidence sets that simultaneously cover the differences $\theta_j-\theta_k$ for all relevant pairs of indices $(j,k)$. While their paper does not explicitly show how these can be constructed when $X_1,\ldots,X_p$ are not independent (which by construction is the case with multinomial data), it is not difficult to propose an appropriate bootstrap procedure (see Appendix~\ref{app: bootstrap}).

The general approach in \cite{Mogstad:2024aa} does not require $X$ to follow a multinomial distribution. The purpose of the remainder of this paper is to show that with this additional distributional assumption it is possible to construct confidence sets for ranks that control the coverage probability not only asymptotically, but in finite samples.  

\begin{remark}[Definition of Rank]\label{rem: set of ranks}
	To simplify the exposition in this remark, suppose we are interested in a single category, $J_0=\{j_0\}$. In the presence of ties, the rank of a category can be defined in different ways. For any $j\in J$, let $\underline r_j \equiv 1 + \sum_{k \in J} \mathds{1}\{\theta_k > \theta_j \}$ and $\bar r_j\equiv p - \sum_{k \in J} \mathds{1} \{\theta_k < \theta_j \}$ be the smallest (i.e., best) and largest (i.e., worst) possible rank of category $j$. If category $j_0$ is not tied with any other category, then $\underline{r}_{j_0}=\bar{r}_{j_0}$ and the rank is unique. On the other hand, when category $j_0$ is tied with at least one other category, then $\underline{r}_{j_0} < \bar{r}_{j_0}$ and different definitions of the rank may select different values from the interval $R_{j_0}\equiv [\underline{r}_{j_0},\bar{r}_{j_0}]$. An inspection of the proof of Theorem~\ref{thm: coverage} reveals that the confidence set $R_{n}$ not only covers our definition of the rank, $r_{j_0}$, in the sense of \eqref{eq: coverage}, but also any other ``reasonable'' definition of the rank in the sense that
	\begin{equation*}
		 P\left\{R_{j_0} \subseteq R_{n}^{\rm cont}\right\} \geq 1 - \alpha,
	\end{equation*}
	where $R_{n}^{\rm cont} \equiv [\min (R_{n}),\max(R_{n})]$ is the interval from the smallest to the largest value in the confidence set $R_{n}$.
\end{remark}

\subsection{Marginal and Simultaneous Confidence Sets for Ranks} \label{sec:margsimul}

In light of the previous discussion, for the construction of a confidence set satisfying \eqref{eq: coverage}, it remains to propose a procedure for testing the family of hypotheses $H_{j,k}$, $(j,k)\in I$, that controls $FWER_I$. In this section, we propose a test of the individual hypothesis $H_{j,k}$ with nominal level $\beta_{j,k}$ and then choose the constants $(\beta_{j,k}\colon (j,k)\in I)$ in a way that controls $FWER_I$ in the sense of \eqref{eq: FWER control}. 

Let $S_{j,k}\equiv X_j + X_k$. One can show that the conditional distribution of $X_j$ given $S_{j,k}=s$ is binomial based on $s$ trials and success probability $\theta_j/(\theta_j+\theta_k)$; a simple proof appears in the proof of Theorem \ref{thm: size control}.
Notice that $H_{j,k}$ is equivalent to $\theta_{j,k} \le 1/2$, where $\theta_{j,k} = \theta_j / ( \theta_j + \theta_k )$.
Conditioning on $S_{j,k}$  eliminates nuisance parameters and reduces the testing problem to a one-parameter problem of testing a binomial probability.  An exact level $\beta_{j,k}$ test may be therefore be easily constructed.
 In particular, the (possibly randomized) test of $H_{j,k}$ defined by the critical function
\begin{equation}\label{eq: def of test}
	\phi(x,s) = \left\{ \begin{array}{cc} 1,& \text{if } x > C(s)\\ \gamma(s), & \text{if } x=C(s)\\ 0,& \text{if } x < C(s)\end{array} \right.
\end{equation}
with constants $\gamma(s)$, $C(s)$ determined by
\begin{equation}\label{eq: def of constants}
	\sum_{i = C(s) + 1}^{s} {s \choose i} (1/2)^{s} + \gamma (s) {s \choose {C(s)}} (1/2)^{s} = \beta_{j,k}\qquad \forall s,
\end{equation}
is an exact level $\beta_{j,k}$ test of $H_{j,k}$.  
The test has rejection probability equal to $\beta_{j,k}$ when $\theta_{j,k} = 1/2$.
Moreover, since the binomial family of distributions has monotone likelihood ratio, the test has rejection probability strictly less than 
$\beta_{j,k}$ whenever $\theta_{j,k} < 1/2$.
In fact, Theorem~\ref{thm: size control} below shows that, for testing $H_{j,k}$, the test $\phi  (X_j, S_{j,k} )$ is
uniformly most powerful  level $\beta_{j,k}$  among all level $\beta_{j,k}$  unbiased tests based on $(X_1 , \ldots , X_p )$; i.e., it is UMPU.
If $\gamma(s)>0$ and one wishes to avoid randomization of the test, then one may simply reject $H_{j,k}$ iff $X_j > C( S_{j,k} )+1$. The $p$-value for this slightly conservative approach of testing $H_{j,k}$ when $S_{j,k}=s$ can be written as
\begin{equation}\label{eq: p-val}
	\hat{p}_{j,k} \equiv  \frac{1}{2^{s}}\sum_{i=X_j}^{s} {s \choose i}.
\end{equation}
The following theorem summarizes the above discussion.

\begin{theorem}\label{thm: size control}
	For any $(j,k)\in J^2$, $j\neq k$, and $\beta_{j,k}\in(0,1)$, the test $\phi ( X_j, S_{j,k} )$ defined by \eqref{eq: def of test} and \eqref{eq: def of constants} is a UMPU level $\beta_{j,k}$ test of $H_{j,k}$.
\end{theorem}

This theorem shows that $\phi$ defines a level $\beta_{j,k}$ test of $H_{j,k}$. To satisfy \eqref{eq: FWER control} one could combine the individual tests, i.e., choose the $(\beta_{j,k}\colon (j,k)\in I)$, by a Bonferroni correction or by the \cite{Holm:1979xy} procedure, for example. Theorem~\ref{thm: coverage} then implies that the confidence set $R_n$, based on such a procedure, has coverage probability no less than $1-\alpha$.

The steps necessary for construction of the proposed confidence sets for the ranks, using the non-randomized test with $p$-value in \eqref{eq: p-val}, are summarized as follows:

\begin{algorithm}\label{alg:marg sim CS} \hspace{1cm}
\vspace{-.2cm}
	\begin{enumerate}
		\item Choose the set $J_0\subseteq J$ of categories of interest.
		\item Set $I$ equal to one of $J^{\rm lower}$, $J^{\rm upper}$, or $J^{\rm two-sided}$, depending on whether lower, upper, or two-sided confidence bounds on the ranks of categories in $J_0$ are desired.
		\item Test the family of hypotheses $H_{j,k}$, $(j,k)\in I$, so that $FWER_I$ is controlled. For instance:
			\begin{itemize}
				\item {\bf Bonferroni:} $H_{j,k}$ is rejected iff
					$$\hat{p}_{j,k} \leq \frac{\alpha}{|I|}. $$

				\item {\bf Holm:} order the p-values $\hat{p}_{j,k}$, $(j,k)\in I$, from the smallest to the largest, $\hat{p}_{(1)} \leq \cdots \leq \hat{p}_{(|I|)}$, and denote the corresponding hypotheses by $H_{(1)},\ldots,H_{(|I|)}$. Then, $H_{(l)}$ is rejected iff 
				$$\hat{p}_{(l')} \leq \frac{\alpha}{|I|+1-l'}\quad \forall l'\leq l. $$
			\end{itemize}
		\item For each $j\in J_0$, collect the rejected hypotheses as in \eqref{eq: rej minus} and \eqref{eq: rej plus}.

		\item Construct $R_n$, the confidence set for the ranks $(r_j\colon j\in J_0)$, as in \eqref{eq: def joint CS}.
	\end{enumerate}
\end{algorithm}

Since the Holm procedure rejects at least as many hypotheses as Bonferroni (with probability one) and thus leads to confidence sets that are at least as short as those based on Bonferroni, it is to be preferred. However, in simulations (Section~\ref{sec: sim}), we find that the two methods lead to almost identical confidence sets, and both methods are optimal in many high-dimensional settings (\citet[Chapter 13.5]{Lehmann:2022aa}).

Theorems~\ref{thm: coverage} and \ref{thm: size control} imply that the resulting confidence set for the ranks is valid in finite samples:

\begin{corollary}
	The confidence set $R_n$ constructed by Algorithm~\ref{alg:marg sim CS} satisfies \eqref{eq: coverage}.
\end{corollary}

One important aspect to note is that we have not imposed any assumptions on the vector of probabilities $\theta$, besides it being the vector of probabilities associated with a multinomial distribution for $p$ categories. In particular, the confidence set $R_n$ satisfies the coverage result \eqref{eq: coverage} regardless of whether any of the elements of $\theta$ are equal to each other (``ties'') or close to each other (``near-ties''). This is an important feature of our confidence sets for the ranks because it ensures that the coverage property does not break down when some categories are observed with equal or similar counts. In contrast, a ``naive'' bootstrap confidence set is valid only in the absence of ties and may substantially under-cover when there are near-ties (see Remark~\ref{rem: naive bootstrap} and the simulations in Section~\ref{sec: sim}).

\begin{remark}[Clopper-Pearson]\label{rem: CP}
	An alternative approach to the construction of confidence sets for the ranks that are valid in finite samples could be based on Clopper-Pearson intervals for binomial probabilities (\cite{Clopper:1934dd}). To see this, note that one could form a Clopper-Pearson interval separately for each element of $\theta$. With a Bonferroni correction one could then combine the marginal confidence intervals into a simultaneous confidence set for the vector $\theta$, which would be valid in finite samples. Given this simultaneous confidence set for $\theta$, one could apply the approach by \cite{Klein:2020oi} to form a simultaneous confidence set for the ranks of all categories, which would also be valid in finite samples.  We compare this approach with ours in the simulations of Section~\ref{sec: sim} and find that the two methods often perform similarly well, but sometimes the Clopper-Pearson intervals are meaningfully wider. One reason why this construction may lead to wide confidence sets is that \cite{Klein:2020oi}'s approach implicitly tests whether two success probabilities are equal by checking whether the Clopper-Pearson intervals for the two success probabilities overlap or not. This construction is excessively crude compared to the use of confidence sets for the differences.
\end{remark}

\subsection{Confidence Sets for the \texorpdfstring{$\tau$}{}-Best} \label{sec:taubest}

Let $R_n \equiv \prod_{j\in J_0} R_{n,j}$ be a simultaneous lower confidence bound on the ranks of all categories, i.e., each $R_{n,j}$ has upper bound equal to $p$ and $R_n$ satisfies \eqref{eq: coverage} for $J_0=J$. Then, the projection
\begin{equation}\label{eq: def proj tau-best}
	R^{\tau-\rm{best}}_n \equiv \left\{j \in J : \tau \in R_{n,j} \right\}
\end{equation}
is a confidence set for the $\tau$-best categories:
\begin{corollary}
	If $R_n \equiv \prod_{j\in J_0} R_{n,j}$ is defined as in Theorem~\ref{thm: coverage} for $J_0=J$ and $I=J^{\rm lower}$, then $R^{\tau-\rm{best}}_n$ as defined in \eqref{eq: def proj tau-best} satisfies \eqref{eq:taubestcoverage}.
\end{corollary}

The construction of such a confidence set using the non-randomized test with $p$-values in \eqref{eq: p-val} can thus be summarized as follows:

\begin{algorithm}\label{alg: tau-best}
	\begin{enumerate}
		\item Set $J_0=J$ and $I=J^{\rm lower}$.
		\item Perform Steps 3--5 of Algorithm~\ref{alg:marg sim CS} to obtain $R_n\equiv \prod_{j\in J_0} R_{n,j}$.
		\item Construct $R^{\tau-\rm{best}}_n$ as defined in \eqref{eq: def proj tau-best}.
	\end{enumerate}
\end{algorithm}

\begin{remark}[$\tau$-Worst]
	A confidence set for the $\tau$-worst categories, $R_0^{\tau-\rm{worst}} \equiv \{j \in J : r_j \geq p-\tau+1 \}$ can be constructed in a similar fashion as in Algorithm~\ref{alg: tau-best} for the $\tau$-best. First, set $J_0=J$ and $I=J^{\rm upper}$. Then perform Steps 3--5 of Algorithm~\ref{alg:marg sim CS} to obtain $R_n\equiv \prod_{j\in J_0} R_{n,j}$. Finally, construct the confidence set $R^{\tau-\rm{worst}}_n \equiv \left\{j \in J : p-\tau+1 \in R_{n,j} \right\}$.
\end{remark}

\section{Bootstrap Confidence Sets}
\label{sec: bootstrap}

As was previously seen, confidence sets for ranks can be based on simultaneous tests of the  hypotheses
$H_{j,k}$ with $(j,k) \in I$ for an appropriate set of indices $I \subset J^2$.  In a similar way, inference for ranks can also be based on simultaneous bootstrap confidence sets  $C_n ( 1- \alpha , I)$ for the differences  $(\theta_j - \theta_k\colon (j,k)\in I)$, which we now describe.

Let $X^*\equiv(X_1^*,\ldots,X_p^*)$ denote a bootstrap draw from the the multinomial distribution with parameters $n$ and $\hat{\theta}\equiv X/n$. Define the bootstrap estimator $\hat{\theta}^*\equiv X^*/n$ and 
\begin{equation*}
	\hat \sigma^*_{j,k} \equiv  \sqrt{\hat{\theta}^*_j ( 1-  \hat{\theta}^*_j ) + \hat{\theta}^*_k (1-  \hat{\theta}^*_k ) + 2 \hat{\theta}^*_j \hat{\theta}^*_k} ~.
\end{equation*}
Consider the bootstrap statistic
$$T_{{\rm lower}, n}^*(I) \equiv \max_{(j,k) \in I} \frac{\hat{\theta}^*_j - \hat{\theta}^*_k - (\hat{\theta}_j-\hat{\theta}_k)}{\hat{\sigma}^*_{j,k} / \sqrt{n} }, $$
where we adopt the convention that $0/0 = 0$ and $c/0 = \text{sign}(c)\infty$ for $c\neq 0$, and denote by $c_{{\rm lower},n}(1-\alpha,I)$ the $(1-\alpha)$-quantile of $T_{{\rm lower}, n}^*(I)$ conditional on the data.\footnote{\label{foot: zeros}In a given bootstrap sample, the ratio inside the max of $T_{{\rm lower},n}^*$ can have a zero denominator and/or zero numerator. For instance, when two categories $j$ and $k$ both have small success frequencies in the data ($\hat{\theta}_j$ and $\hat{\theta}_k$ are small), then it is possible that a given bootstrap sample does not contain any success for either of the two categories, i.e., $\hat{\theta}_j^*=\hat{\theta}_j^*=\hat{\sigma}_{j,k}^* = 0$, and the denominator is zero. When $\hat{\theta}_j<\hat{\theta}_k$, then the resulting critical value $c_{{\rm lower},n}(1-\alpha,I)$ equals $\infty$. On the other hand, when the frequencies in the data are equal ($\hat{\theta}_j=\hat{\theta}_k$), then both the numerator and denominator of the ratio are zero and the maximum is determined by other bootstrap samples that produce a positive ratio.} We can then construct lower confidence bounds for the vector of differences $\Delta_I\equiv (\theta_j-\theta_k\colon (j,k)\in I)$ by
\begin{equation*}
	C_{{\rm lower}, n}(1 - \alpha,I) \equiv \prod_{(j,k) \in I} C_{{\rm lower}, n,j,k}(1 - \alpha,I)
\end{equation*}
with
$$C_{{\rm lower}, n,j,k}(1 - \alpha,I) \equiv \Biggr [ \hat{\theta}_j - \hat{\theta}_k - c_{{\rm lower}, n}(1- \alpha,I) \frac{ \hat{\sigma}_{j,k}}{\sqrt{n}}, \infty \Biggr ) $$
and $\hat \sigma_{j,k} \equiv  \sqrt{\hat{\theta}_j ( 1-  \hat{\theta}_j ) + \hat{\theta}_k (1-  \hat{\theta}_k ) + 2 \hat{\theta}_j \hat{\theta}_k}$. As long as all $\theta_j$, $j\in J$, are nonzero, this confidence set covers the vector of true differences with probability $1-\alpha$, asymptotically as the sample size $n$ tends to infinity:
\begin{equation}\label{eq: CS diff coverage}
	\lim_{n\to\infty} P\{\Delta_I \in C_{{\rm lower},n}(1-\alpha,I) \} = 1-\alpha.
\end{equation}
Appendix~\ref{app: bootstrap} provides a formal justification of this claim and further shows that the coverage probability is no less than $1-\alpha$ when some $\theta_j=0$.  Let $I$ be equal to one of the sets $J^{\rm lower}$, $J^{\rm upper}$, or $J^{\rm two-sided}$ depending on which type of confidence set for the ranks is desired. Consider the test that rejects $H_{j,k}$ iff $C_{{\rm lower},n,j,k}(1-\alpha,I)$ lies entirely above zero. Then, based on this test, form the sets $\text{Rej}_{j}^{-}$ and $\text{Rej}_{j}^{+}$ as in \eqref{eq: rej minus} and \eqref{eq: rej plus}. The bootstrap confidence set for the ranks of categories in $J_0$ can then be constructed as in \eqref{eq: def joint CS}; denote the resulting confidence set by $R_n^{\rm boot} \equiv \prod_{j\in J_0} R_{n,j}^{\rm boot}$. By an argument similar to that in Theorem~3.3 in \cite{Mogstad:2024aa}, the probability that this confidence set covers the true ranks is bounded from below by the probability that the vector of differences, $\Delta_I$, is covered by $C_{\rm lower, n}(1 - \alpha,I)$. Therefore, the validity of the bootstrap in the sense of \eqref{eq: CS diff coverage} implies that the bootstrap confidence set for the ranks also covers the true ranks with probability at least $1-\alpha$ in the limit as $n\to \infty$. The following result formalizes this discussion and shows that the validity of the bootstrap does not require any further assumptions:

\begin{theorem}\label{thm: bootstrap validity}
	For $R_n^{\rm boot}$ defined in the previous paragraph, we have
	\begin{equation} \label{eq: asymptotic coverage}
	 \liminf_{n\to\infty} P\left\{r_{j} \in R_{n,j}^{\rm boot}\;\forall j\in J_0\right\} \geq 1 - \alpha.
	\end{equation}
\end{theorem}

\begin{remark}[Exactness]
It is possible to show that there exists $\theta$ such that
$$\lim_{n \rightarrow \infty} P\left\{R_j \in R_{n,j}^{\rm boot}\;\forall j\in J_0\right\} = 1 - \alpha,$$ where $R_j$ is defined as in Remark \ref{rem: set of ranks}.  In this sense, the bootstrap-based confidence sets described above are non-conservative.
\end{remark}

\begin{remark}[Uniformity]
	The validity of the bootstrap confidence set for the vector of differences as in \eqref{eq: CS diff coverage} also holds uniformly over data-generating processes in certain classes of distributions. Such a statement could be established using the results in \cite{Romano:2012kj}. One important assumption for the applicability of their results to our setting is that the elements of $\theta$ need to be bounded away from $0$ and $1$. As in \cite{Mogstad:2024aa}, uniform validity of the confidence sets for the differences then implies uniform validity of the confidence sets for the ranks.
\end{remark}

\begin{remark}[Stepwise Improvements]\label{rem: stepwise}
	One could use stepdown procedures from \cite{romano2005gf} to improve the confidence sets for the ranks. See \cite{Mogstad:2024aa} for more details.
\end{remark}

\begin{remark}[Two-sided Confidence Sets]\label{rem: two-sided}
	Suppose the goal is to construct a rectangular two-sided confidence set for the ranks of categories in $J_0$. Instead of testing the one-sided hypotheses $H_{j,k}$ for all pairs in the large set $J^{\rm two-sided}$ with the one-sided confidence sets for the differences, one could also test whether the differences are zero using a smaller number of two-sided confidence sets for the differences.

	To see this let $C_{{\rm symm}, n}(1 - \alpha,I) \equiv \prod_{(j,k) \in I} C_{{\rm symm}, n,j,k}(1 - \alpha,I)$ with 
	\begin{align*}
		C_{{\rm symm}, n,j,k}(1 - \alpha, I) &\equiv \Biggr [\hat{\theta}_j - \hat{\theta}_k \pm c_{{\rm symm}, n}(1- \alpha,I) \frac{ \hat{\sigma}_{j,k}}{\sqrt{n}} \Biggr ],
	\end{align*}
	where $c_{{\rm symm},n}(1-\alpha,I)$ denotes the $(1-\alpha)$-quantile of 
	$$T_{{\rm symm}, n}^*(I) \equiv \max_{(j,k) \in I} \frac{\left|\hat{\theta}^*_j - \hat{\theta}^*_k - (\hat{\theta}_j-\hat{\theta}_k)\right|}{\hat{\sigma}^*_{j,k} / \sqrt{n} }$$
	conditional on the data. Then, set $I=J^{\rm upper}$ and compute
	\begin{align*}
		N_{j}^{-} &\equiv \{k\in J\setminus \{j\}\colon C_{{\rm symm}, n,j,k}(1 - \alpha, I) \text{ lies entirely below zero }\}\\
		N_{j}^{+} &\equiv \{k\in J\setminus \{j\}\colon C_{{\rm symm}, n,j,k}(1 - \alpha, I) \text{ lies entirely above zero }\}
	\end{align*}
	which indicate the categories $k$ with probabilities strictly larger or smaller than $\theta_j$. A confidence set for the ranks of categories in $J_0$ can then be formed as in \eqref{eq: def joint CS}, replacing $\text{Rej}_j^-$ and $\text{Rej}_j^+$ by $N_j^-$ and $N_j^+$, respectively. By arguments analogous to those for the one-sided confidence sets, the resulting confidence set for the ranks of categories in $J_0$ then covers the true ranks with probability approaching at least $1-\alpha$ as $n\to\infty$.
\end{remark}

\begin{remark}[Studentization]\label{rem: stud}
	The bootstrap procedure in Remark~\ref{rem: two-sided} may perform poorly in the sense of under-covering the true ranks when there are many categories and all estimated probabilities $\hat{\theta}_1,\ldots,\hat{\theta}_p$ are small. In such situations, the ratio in the definition of the bootstrap statistic may evaluate to $0/0$ on many bootstrap samples, leading to a critical value that is too small. In addition, the bootstrap procedure may perform poorly in the sense of yielding confidence sets that are very wide when there are two or more categories with small estimated probabilities. In such situations, there may be bootstrap samples without any successes for categories $j$ and $k$, so the ratio in the definition of the bootstrap statistic evaluates to $\infty$ and the resulting critical value is equal to $\infty$; see Footnote~\ref{foot: zeros}.  For these reasons, it may be beneficial not to studentize $T^*_{{\rm symm},n}(I)$, in which case one would also remove $\hat{\sigma}_{j,k}$ from the expression of $C_{{\rm symm},n,j,k}(1-\alpha,I)$.  These two approaches are compared further in the simulations in Section~\ref{sec: sim largep}.  
\end{remark}

\begin{remark}[``Naive'' Bootstrap]\label{rem: naive bootstrap}
	Suppose the goal is to construct a confidence set for the rank of a single category. The confidence set $R_n$ based on Algorithm~\ref{alg:marg sim CS} was shown to be valid in finite samples, regardless of the value of the vector $\theta$. In particular, there may be an arbitrary number of ties or near-ties in $\theta$. Similarly, the confidence set $R_n$ based on the bootstrap as proposed in this section is asymptotically valid regardless of the number of ties or near-ties in $\theta$. On the other hand, the bootstrap as proposed by, e.g., \cite{Goldstein:1996re} performs poorly when, for some $k \neq j$, $\theta_k$ is (close to) equal to $\theta_j$. For concreteness, consider the following ``naive'' bootstrap procedure. For a category $j$, denote by $\hat{\theta}_j^*$ the estimator of $\theta_j$ computed on a bootstrap sample and let $\hat{r}_j^*$ be the rank computed using the bootstrap estimators $\hat{\theta}_1^*,\ldots,\hat{\theta}_p^*$. Confidence sets for $r_j$ could then be constructed using upper and/or lower empirical quantiles of $\hat{r}_j^*$ conditional on the data. \cite{Mogstad:2024aa} show that this intuitive approach fails to deliver the desired coverage property when there are ties (unless $p=2$). In fact, the coverage probability tends to zero as $p$ grows. For further discussion, see \cite{Xie:2009oi} and \cite{Hall:2009oi}. In contrast, our bootstrap approach does not rely on a consistent estimator of the distribution of estimated ranks but rather on the availability of simultaneous bootstrap confidence sets for the differences $\Delta_I$ with asymptotic coverage no less than the desired level.  Such simultaneous confidence sets are available under weak conditions and, in particular, do not restrict the configuration of the vector of probabilities $\theta$. In comparison to \cite{Xie:2009oi}, our bootstrap procedure also circumvents smoothing of the indicator in the definition of the ranks and thus the need for choosing such a smoothing parameter.
\end{remark}

\section{Ranking Political Parties by Voters' Support in the Australian Election Study 2019}
\label{sec: emp}
In this section, we apply the inference procedures from Sections~\ref{sec:main} and \ref{sec: bootstrap} to examine the ranking of political parties by their share of voters' support in the Australian Election Study (AES). The AES has fielded representative surveys after every federal election since 1987 and provides the most comprehensive source of evidence on political attitudes in Australia \citep{cameronAES}. We use AES data from 2019 with address-based stratified random sampling from the Geocoded National Address File (G-NAF)\citep{beanAES}.\footnote{The original sampling methodology description reads: ``Within the parameters outlined above, the new AES sample was selected from the G-NAF database using a stratified sample design in accordance with the geographical distribution of the Australian residential population aged 18 years and over. GNAF sample selections were supplied by the MasterSoft Group. A total of 3,944 sample records were randomly generated within 15 geographic strata (see Table 2) to ensure sufficient sample was utilised to achieve the desired number of responses for the AES"\citep{beanAES}. We interpret this sampling as being i.i.d. within territories.} Table~\ref{tab: DescrStat} shows a total of 3,944 sampled eligible voters from all fifteen Australian territories resulting in 1,211 respondents. In the subsequent analysis we work with respondents and refer to them as ``sample". 

To examine which political parties are on the top and the bottom of the ranking in each Australian territory, we use respondents' answers to the survey question ``Generally speaking, do you usually think of yourself as Liberal, Labor, National, Greens or other (specify)?". The answer categories include political parties; ``Skipped" and ``No answer" categories that we group in a single ``No answer" category; ``Independent", ``Swing Voter" and ``No party" categories that we group in a single ``No party" category. For more populous territories, we observe between seven and ten categories with positive support shares.

By applying the inference procedures from Sections~\ref{sec:main} and \ref{sec: bootstrap} we compute (i) the marginal confidence set for the rank of a particular political party and (ii) the simultaneous confidence set for the ranks of all parties. Thus, (i) is relevant if one is interested whether a particular party is on the top (bottom) of ranking by the voters' support, and (ii) is relevant if one is interested in the entire ranking of parties.

\begin{table}[t]
\centering
\begin{tabular}{lccc}
   \midrule
Territory & Sampled Voters & Respondents & Categories \\ 
   \midrule
Greater Sydney & 816 & 238 & 8 \\ 
  Greater Melbourne & 780 & 234 & 7 \\ 
  Rest of New South Wales & 445 & 144 & 8 \\ 
  Rest of Queensland & 402 & 121 & 10 \\ 
  Greater Brisbane & 378 & 115 & 8 \\ 
  Greater Perth & 319 & 93 & 7 \\ 
  Rest of Victoria & 248 & 82 & 8 \\ 
  Greater Adelaide & 217 & 81 & 9 \\ 
  Rest of Western Australia & 87 & 26 & 7 \\ 
  Australian Capital Territory & 67 & 24 & 4 \\ 
  Rest of South Australia & 63 & 17 & 4 \\ 
  Rest of Tasmania & 47 & 16 & 3 \\ 
  Greater Hobart & 35 & 12 & 4 \\ 
  Greater Darwin & 24 & 6 & 3 \\ 
  Rest of Northern Territories & 16 & 2 & 2 \\ 
   \midrule
Total: & 3,944 & 1,211 &  \\ 
   \bottomrule
\end{tabular}
\caption{Australian Election Study 2019 G-NAF 
    stratified random sample of 3,944 eligible
    voters resulting in 1,211 respondents.
    In the subsequent analysis we work with 
    respondents and refer to them as ``sample". 
    The last column shows the number of categories
    with positive support share in each territory 
    measured by answers to AES2019 survey question 
    ``Generally speaking, do you usually think of 
    yourself as Liberal, Labor, National, Greens 
    or other(specify)?". The answers include 
    political parties, ``Skipped" and ``No answer" 
    categories that we group in a single ``No answer" 
    category; ``Independent", ``Swing Voter" and 
    ``No party" categories that we group in a single
    ``No party" category.} 
\label{tab: DescrStat}
\end{table}

\subsection{Marginal Confidence Sets for the Ranks of Political Parties in Greater Melbourne }
\label{sec: Melbourne}

Consider first a particular territory, Greater Melbourne. Figure~\ref{fig: Melbourne} shows the point estimates and standard errors for voters' support share for each category with positive number of supporters. The leftmost panel in row A shows considerable variation in point estimates across categories from 0.371 for the most supported (Labor Party) to 0.004 for the least supported (National Party). Support shares on the top and the bottom of the ranking are close to each other, while the shares in the middle are better separated.

The middle panel in row A of Figure~\ref{fig: Melbourne} presents the 95\% marginal confidence sets for the rank of each category implemented using five procedures: the exact Holm (\textbf{``exactHolm"}) described in the Algorithm~\ref{alg:marg sim CS}, Clopper-Pearson (\textbf{``CP"}) as in Remark~\ref{rem: CP}, non-studentized (\textbf{``boot"})  and studentized (\textbf{``bootStud"}) versions of the bootstrap as in Section~\ref{sec: bootstrap} and the ``naive" bootstrap (\textbf{``naive"}) as in Remark~\ref{rem: naive bootstrap}. The first two methods have been shown to be valid in finite samples, and the studentized and non-studentized bootstrap are motivated by asymptotic validity. The ``naive" bootstrap is asymptotically valid in the absence of (near-)ties (see Remark~\ref{rem: naive bootstrap}), but invalid otherwise.

 The resulting confidence sets exhibit four pronounced patterns. First, the ``naive" bootstrap confidence sets are the tightest. As indicated in Remark~\ref{rem: naive bootstrap}, the ``naive" bootstrap produces confidence sets that fail to cover the true ranks with the desired probability when there are (near-) ties. Our simulations in Section~\ref{sec: sim} confirm that, in datasets like the one from Greater Melbourne, the ``naive" bootstrap does indeed produce short confidence sets at the expense of its coverage frequency lying substantially below the desired nominal level. Second, the exact Holm procedure produces weakly shorter confidence sets than Clopper-Pearson and the studentized and non-studentized bootstraps. For example, the exact Holm confidence sets for the ranks of the categories ``No party" and ``Greens'' contain only ranks three and four, respectively, while Clopper-Pearson and the studentized and non-studentized bootstraps produce confidence sets containing at least two ranks. Third, both the studentized and non-studentized bootstrap confidence sets are wide in the middle of the ranking. Their length for the 4th category, ``Greens", is equal to four compared to the maximum possible length of six. Fourth, the studentized bootstrap produces extremely wide confidence sets at the bottom of the ranking for the last two categories. These confidence sets cover the entire ranking and correspond to infinite critical values. In contrast, the length of the finite-sample valid confidence sets for the last two categories is two. These patterns suggest that the finite-sample valid confidence sets for parties in Greater Melbourne are informative, and the exact Holm procedure is the most informative among the valid procedures (i.e., excluding the ``naive" bootstrap). 

As discussed in Footnote~\ref{foot: zeros}, the studentized bootstrap confidence sets are wide when the ratio in the bootstrap test statistic has a denominator that is equal to zero and a positive numerator in a substantial fraction of the bootstrap samples. This circumstance arises when at least two categories have small but positive shares in the data so that, in the bootstrap samples, $\hat{\theta}_j^* = \hat{\theta}_k^* = \hat{\sigma}_{j,k}^* = 0$ while, in the data, $|\hat{\theta}_j - \hat{\theta}_k| > 0$. One solution is to group the categories with small shares together. The leftmost panel in row B of Figure~\ref{fig: Melbourne} shows support shares when we group ``National Party", ``One Nation" and ``No answer" into a single category ``Other". The middle panel in row B shows that in the middle or at the bottom of the ranking, both the studentized and non-studentized bootstraps no longer produce confidence sets as wide as in row A. Notably, the exact Holm confidence sets are still tighter than Clopper-Pearson confidence sets. 

An alternative solution to the division by zero in bootstrap samples is to reduce the confidence level. The rightmost column of Figure~\ref{fig: Melbourne} shows the 90\% marginal confidence sets for the ranks using both the original set of categories in row A and with the three smallest categories grouped in row B. Indeed, with lower confidence level we no longer observe studentized bootstrap confidence sets covering the entire ranking. However, both finite-sample methods still produce weakly smaller confidence sets than both types of bootstrap. Furthermore, the panel with grouped small categories in row B shows that the exact Holm confidence sets for ``Greens" are less informative than for the original categories, regardless of whether the confidence level is 90\% or 95\%.

\begin{sidewaysfigure}
\centering
\includegraphics[width =0.93\textwidth]{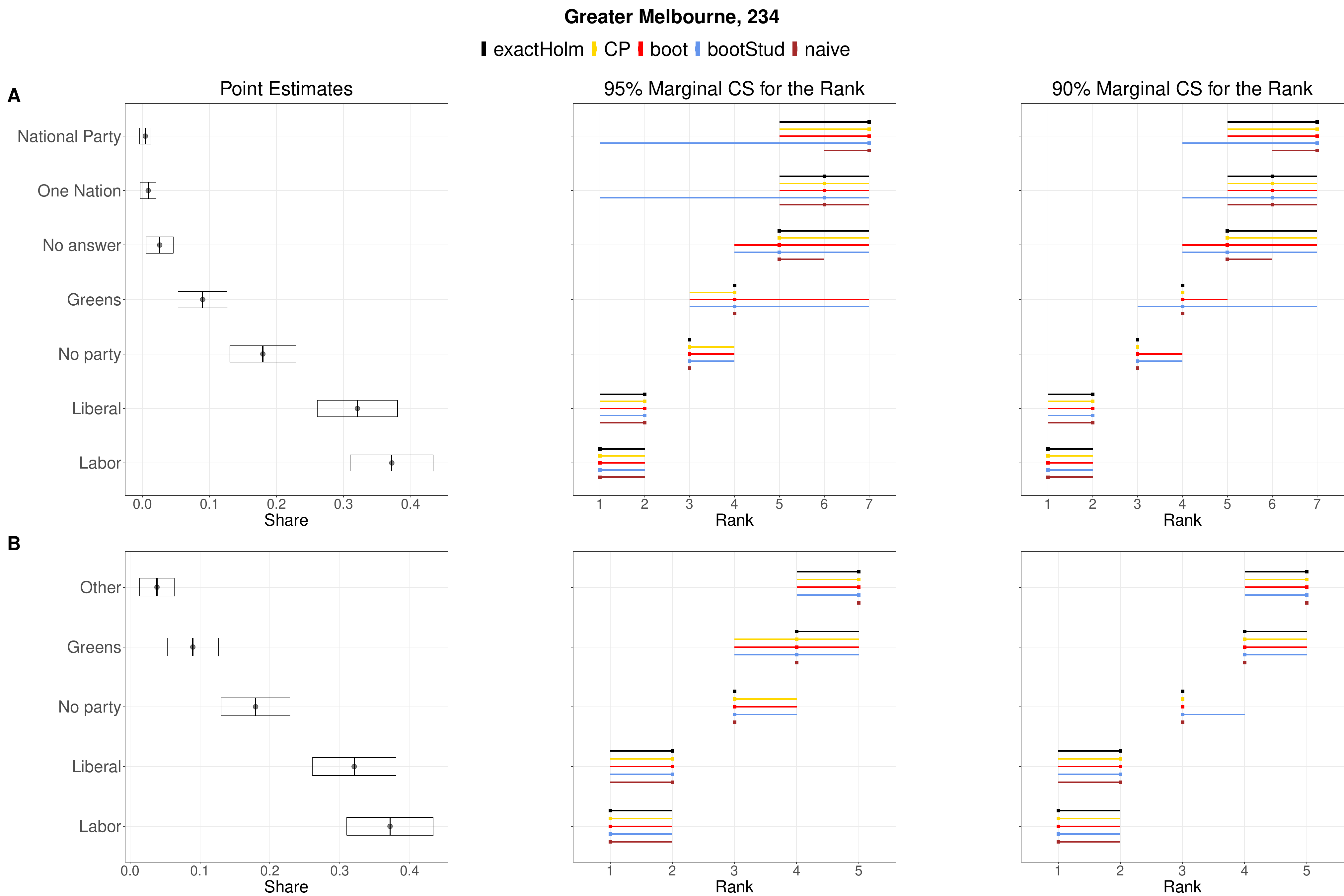}
\caption{\textbf{Left} column shows point estimates of categories' support shares in Greater Melbourne and $\pm 1.96$se. \textbf{Middle} column shows the 95\% marginal confidence sets for the rank of each category computed by five procedures for each category. \textbf{Right} column shows the 90\% marginal confidence sets for the rank of each category computed by five procedures for each category. In row \textbf{A} we use originally defined choice categories. In row \textbf{B}, we group ``National Party", ``One Nation" and ``No answer" into a single category ``Other".}
    \label{fig: Melbourne}
\end{sidewaysfigure}

\subsection{Marginal Confidence Sets for the Eight Most Populous Territories}
Next, we move beyond the example of Greater Melbourne and consider all fifteen Australian territories.  Appendix Figures~\ref{fig: Data1} and \ref{fig: Data2} show the point estimates for support shares in all territories, and Appendix Figures~\ref{fig: MargCS1} and \ref{fig: MargCS2} show the 95\% marginal confidence sets for the rank of each category in each territory computed using the same five procedures as for Greater Melbourne. Similar to our illustration for Greater Melbourne, both bootstrap procedures produce wide confidence sets in the middle and at the bottom of the ranking in the majority of populous territories. Due to small sample sizes the seven least populous territories have mostly uninformative confidence sets for all categories. Therefore, we focus our analysis on the eight most populous territories.

Figure~\ref{fig: MargCS1} shows that none of the valid methods produce tighter confidence sets than all other valid methods uniformly across all categories in all territories. For example, the finite-sample valid confidence sets are weakly tighter than both types of bootstrap confidence sets for each category in Greater Melbourne. In comparison, the bootstrap confidence sets are weakly tighter in the top part of the ranking in Greater Sydney. We summarize this finding in the top panel of Table~\ref{tab: TabCompare} that shows the percentage of category$\times$territory cases across the eight most populous territories where each method produces strictly wider 95\% marginal confidence sets for the rank than other methods. For example, the first row shows that the exact Holm confidence sets are strictly wider than Clopper-Pearson confidence sets in 6.2\% of category$\times$territory cases, strictly wider than the non-studentized bootstrap confidence sets in 12.3\% of cases, and strictly wider than the studentized bootstrap confidence sets in 1.5\% of cases. The first column shows that Clopper-Pearson confidence sets are strictly wider than the exact Holm confidence sets in 3.1\% of category$\times$territory cases, the non-studentized bootstrap confidence sets are strictly wider than the exact Holm in 29.2\% of cases, and the studentized bootstrap confidence sets are strictly wider than the exact Holm confidence sets in 47.7\% of cases. Notably, the share of cases where the studentized and non-studentized bootstrap confidence sets are strictly wider than Clopper-Pearson or the exact Holm confidence sets is substantially larger than the share of cases where Clopper-Pearson and the exact Holm confidence sets are strictly wider than both types of bootstrap confidence sets.
 
 \begin{table}[t]
\centering
\renewcommand{\arraystretch}{1.5}
\begin{tabular}{>{\arraybackslash}p{20mm}>{\centering\arraybackslash}p{20mm}>{\centering\arraybackslash}
>{\centering\arraybackslash}p{20mm}>{\centering\arraybackslash}p{20mm}>{\centering\arraybackslash}p{20mm}
}

\multicolumn{5}{c}{\textbf{Percentage of Wider 95\% Marginal CS, Row vs Column}}\\
\toprule
\multicolumn{5}{c}{\textbf{Original Set of Categories}}\\
\toprule
& exactHolm & CP & boot & bootStud   \\ 
  \midrule
  exactHolm  & \cellcolor{gray} & 6.2 & 12.3 & 1.5  \\ 
  CP & 3.1 & \cellcolor{gray} & 9.2 & 0.0  \\ 
  boot  & 29.2 & 27.7 & \cellcolor{gray} & 0.0  \\ 
  bootStud  & 47.7 & 49.2 & 41.5 & \cellcolor{gray} \\ 
    \bottomrule
\end{tabular}

\begin{tabular}{>{\arraybackslash}p{20mm}>{\centering\arraybackslash}p{20mm}>{\centering\arraybackslash}p{20mm}>{\centering\arraybackslash}p{20mm}>{\centering\arraybackslash}p{20mm}>{\centering\arraybackslash}p{20mm}
}

\multicolumn{5}{c}{\textbf{Small Categories Grouped}}\\
\toprule & exactHolm & CP & boot & bootStud  \\ 
  \midrule
exactHolm & \cellcolor{gray} & 2.5 & 22.5 & 2.5  \\ 
  CP  & 10.0 &\cellcolor{gray} & 22.5 & 0.0  \\ 
  boot& 7.5 & 0.0 & \cellcolor{gray} & 0.0 \\ 
  bootStud  & 15.0 & 5.0 & 27.5 & \cellcolor{gray} \\ 
    \bottomrule
\end{tabular}

\caption{Each cell shows the percentage of pairwise comparisons across all categories in the eight most populous territories where the inference procedure in a row produces wider 95\% marginal confidence sets for the ranks than the procedure in a column. The \textbf{top panel} shows results for the original set of categories in each territory. The \textbf{bottom panel} shows results when we group all categories except ``Liberal", ``Labor", ``Greens" and ``No party" into a single category ``Other".} 
\label{tab: TabCompare}
\end{table}
 
 The bottom panel of Table~\ref{tab: TabCompare} shows the same comparisons when we group all categories except ``Liberal", ``Labor", ``Greens" and ``No party" into a single category ``Other". As discussed above, this grouping prevents zeroes in the bootstrap test statistic denominator and excessively wide bootstrap confidence sets. As a result, the percentage of cases in rows 3 and 4 where both types of bootstrap confidence sets are strictly wider than the finite-sample valid confidence sets decreases. Interestingly, the percentage of cases where the exact Holm confidence sets are strictly wider than Clopper-Pearson confidence sets is lower with grouping, and the percentage of cases where Clopper-Pearson confidence sets are wider than the exact Holm confidence sets is larger.

\subsection{Marginal Versus Simultaneous Confidence Sets}
The analysis of marginal confidence sets answers questions about the rank of a particular party, but one may be interested in the ranking of all parties described by simultaneous confidence sets. Figure~\ref{fig: MargSimul} compares 95\% marginal confidence sets for the ranks to 95\% simultaneous confidence sets for the ranks produced by the exact Holm, Clopper-Pearson, the studentized and non-studentized bootstrap procedures for categories in Greater Melbourne. Naturally, simultaneous confidence sets are weakly wider than marginal confidence sets for each procedure. This feature is more pronounced for the studentized bootstrap confidence sets due to infinite critical values in the bootstrap test statistic for categories at the bottom of the ranking. In contrast, the finite-sample valid simultaneous confidence sets are still informative, and the exact Holm produces weakly tighter confidence sets than all other procedures. Appendix Figure~\ref{fig: MargSimul10} shows that with confidence level reduced to 90\% we no longer observe studentized bootstrap confidence sets as wide as on Figure~\ref{fig: MargSimul} since the critical value becomes finite.

The top panel of Table~\ref{tab: TabCompareSimul} shows that our findings hold in the eight most populous Australian territories with the original set of choice categories. The studentized bootstrap almost always produces strictly wider 95\% simultaneous confidences sets than other methods and never produces tighter confidence sets. The exact Holm simultaneous confidence sets are never strictly wider than Clopper-Pearson or the studentized bootstrap confidence sets, and are strictly wider than the non-studentized bootstrap confidence sets in only 9.2\% of category$\times$territory cases. In contrast, both Clopper-Pearson and the non-studentized bootstrap simultaneous confidence sets are strictly wider than the exact Holm confidence sets in 26.2\% of category$\times$territory cases.

 In the bottom panel of Table~\ref{tab: TabCompareSimul} we group all categories except ``Liberal", ``Labor", ``Greens" and ``No party" into a single category ``Other". As a result, both types of bootstrap confidence sets become tighter. In particular, the non-studentized bootstrap simultaneous confidence sets are never strictly wider than the confidence sets produced by any other inference procedure. The exact Holm confidence sets, however, are still never wider than Clopper-Pearson or the studentized bootstrap confidence sets and are strictly wider than the non-studentized bootstrap confidence sets in only 5\% of category$\times$territory cases.

\begin{figure}[H]
    \centering
    \includegraphics[width=0.99\textwidth]{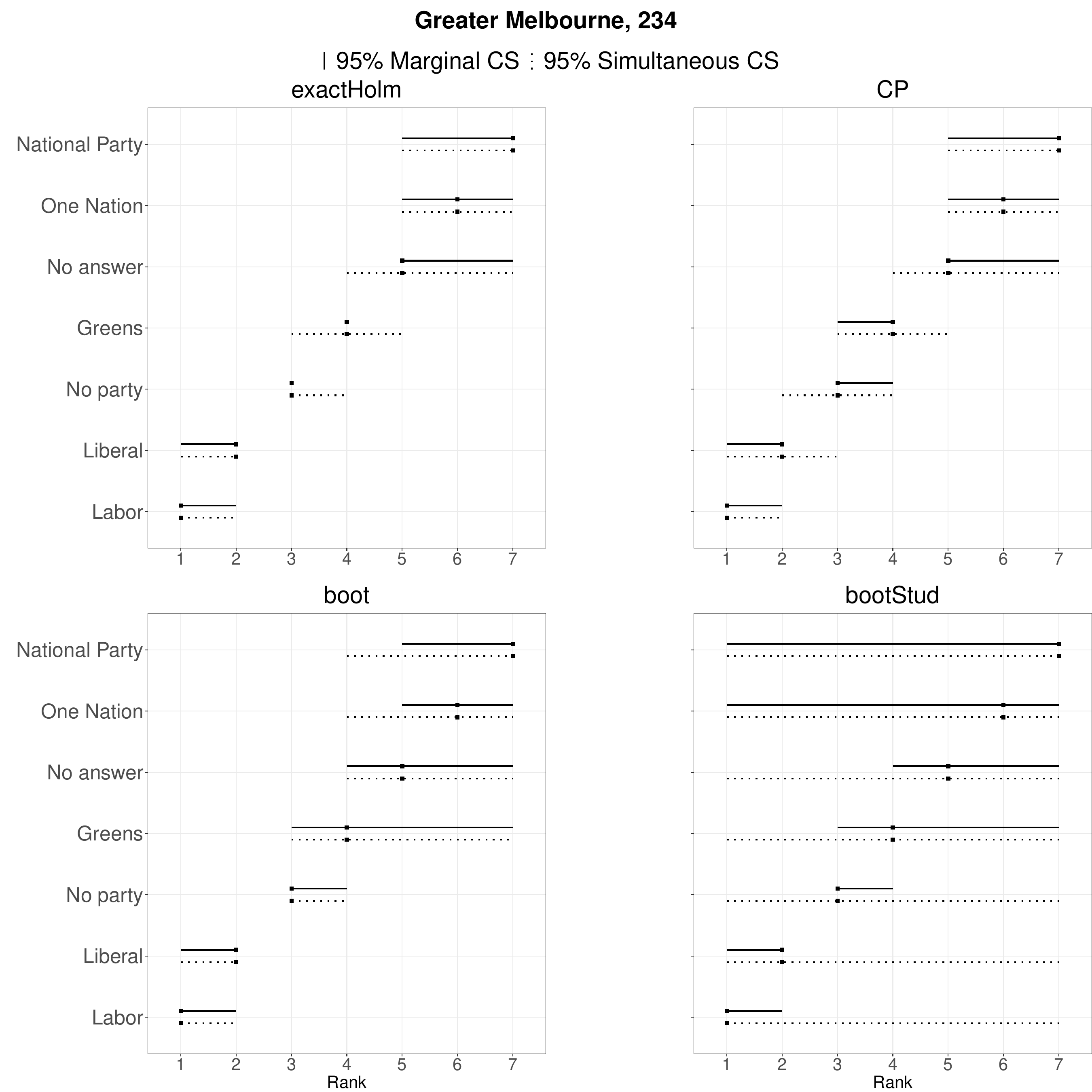}
    \caption{ 95\% marginal and 95\% simultaneous confidence sets for the ranks of categories in Greater Melbourne.  }
    \label{fig: MargSimul}
\end{figure}

\begin{table}[t]
\centering
\renewcommand{\arraystretch}{1.5}
\begin{tabular}{>{\arraybackslash}p{20mm}>{\centering\arraybackslash}p{20mm}>{\centering\arraybackslash}p{20mm}
>{\centering\arraybackslash}p{20mm}>{\centering\arraybackslash}p{20mm}>{\centering\arraybackslash}p{20mm}
}

\multicolumn{5}{c}{\textbf{Percentage of Wider 95\% Simultaneous CS, Row vs Column}}\\
\toprule
\multicolumn{5}{c}{\textbf{Original Set of Categories}}\\
\toprule
& exactHolm & CP & boot & bootStud   \\ 
  \midrule
  exactHolm  & \cellcolor{gray} & 0.0 & 9.2 & 0.0  \\ 
  CP & 26.2 & \cellcolor{gray} & 23.1 & 0.0  \\ 
  boot  & 26.2 & 21.5 & \cellcolor{gray} & 0.0  \\ 
  bootStud  & 87.7 & 83.1 & 76.9 & \cellcolor{gray} \\ 
    \bottomrule
\end{tabular}

\begin{tabular}{>{\arraybackslash}p{20mm}>{\centering\arraybackslash}p{20mm}>{\centering\arraybackslash}p{20mm}
>{\centering\arraybackslash}p{20mm}>{\centering\arraybackslash}p{20mm}>{\centering\arraybackslash}p{20mm}
}

\multicolumn{5}{c}{\textbf{Small Categories Grouped}}\\
\toprule & exactHolm & CP & boot & bootStud  \\ 
  \midrule
exactHolm & \cellcolor{gray} & 0.0 & 5.0 & 0.0  \\ 
  CP  & 20.0 & \cellcolor{gray} & 25.0 & 0.0  \\ 
  boot& 0.0 & 0.0 & \cellcolor{gray} & 0.0 \\ 
  bootStud & 32.5 & 20.0 & 37.5& \cellcolor{gray} \\ 
    \bottomrule
\end{tabular}

\caption{Each cell shows the percentage of pairwise comparisons across all categories in the eight most populous territories where the inference procedure in a row produces wider 95\% simultaneous confidence sets for the ranks than the procedure in a column . The \textbf{top panel} shows results for the original set of categories in each territory. The \textbf{bottom panel} shows results when we group all categories except ``Liberal", ``Labor", ``Greens" and ``No party" into a single category ``Other".} 
\label{tab: TabCompareSimul}
\end{table}

\section{Simulations}
\label{sec: sim}

In this section, we examine the finite-sample performance of the following approaches to constructing confidence sets for ranks:

\begin{description}
	\item {\bf ``exactBonf'':} the confidence set computed through Algorithm~\ref{alg:marg sim CS} using the Bonferroni correction.
	\item {\bf ``exactHolm'':} the confidence set computed through Algorithm~\ref{alg:marg sim CS} using the Holm correction.
	\item {\bf ``CP'':} the confidence set based on Clopper-Pearson confidence sets for binomial probabilities as described in Remark~\ref{rem: CP}.
	\item {\bf ``boot'':} the bootstrap (not studentized) confidence set based on two-sided confidence sets for the differences as described in Remark~\ref{rem: stud}.
	\item {\bf ``bootStud'':} the bootstrap (studentized) confidence set based on two-sided confidence sets for the differences as described in Remark~\ref{rem: two-sided}.
	\item {\bf ``naive'':} the ``naive'' bootstrap confidence set as described in Remark~\ref{rem: naive bootstrap}.
\end{description}

All simulations are based on $1,000$ Monte Carlo samples and nominal coverage of $95\%$. Bootstrap confidence sets are based on $10,000$ bootstrap samples except in Section~\ref{sec: erratic}, where for computational reasons we use $1,000$ bootstrap samples. Coverage is defined as in Remark~\ref{rem: set of ranks}, i.e., coverage of the set of possible ranks $R_j$.

We consider three different designs, starting with one that is calibrated to the dataset on which the empirical application is based. Second, we investigate whether the confidence sets exhibit erratic behavior in coverage frequencies similar to that reported for confidence sets for binomial proportions (\cite{Brown:2001p5787}). The first two designs consider only data generating processes with three or seven categories to be ranked. In the final simulation design, we analyze the behavior of the confidence sets as the number of categories increases.

\subsection{AES Design}

The simulation design in this subsection is calibrated to the AES data for Greater Melbourne as in Section~\ref{sec: Melbourne}. The estimated vector of success probabilities is 
$$\hat{\theta}_{AES}=(0.372, 0.321, 0.179, 0.090, 0.026, 0.009, 0.004)$$
 and the number of respondents is $n_{AES}=234$. The vector of success probabilities employed in the simulations, $\theta$, is parametrized as
$$\theta = (1-\kappa)\frac{1}{p}\iota + \kappa \hat{\theta}_{AES}, $$
where $\iota$ denotes a vector of ones and $\kappa \in [0,1]$. So, when $\kappa=1$, then the vector of probabilities is the same as in the data set. When $\kappa=0$, then all probabilities are equal, and values of $\kappa$ between $0$ and $1$ generate probabilities between the two extremes. A parameter $\tau\in \{0.5,1,2\}$ is introduced to vary the sample size as $n = \tau\, n_{AES}. $ So, when $\tau=1$, then the sample size in the simulation is equal to the one in the data set, but we also consider half and double that sample size.

We begin by recalling the four main findings about 95\% marginal confidence sets for the ranks of categories in Greater Melbourne from Section~\ref{sec: Melbourne}:

\begin{enumerate}
\item naive bootstrap confidence sets are the tightest
\item exactHolm confidence sets are weakly tighter than confidence sets produced by all other valid procedures (i.e. all except the naive bootstrap)
\item boot and bootStud confidence sets are very wide in the middle of the ranking
\item bootStud confidence sets are very wide at the bottom of the ranking
\end{enumerate}

Below we explore each of these findings in depth by focusing on lengths and empirical coverage frequencies of confidence sets for the rank of the 1st (top of the ranking), 4th (middle of the ranking) and 7th (bottom of the ranking) categories in Table~\ref{tab: TabLengthCoverage}. 

First, naive bootstrap confidence sets are the tightest for all three categories in all simulations. Table~\ref{tab: TabLengthCoverage} shows that this tightness comes at the cost of severe under-coverage for the 1st and the 7th categories. For the 7th category naive bootstrap confidence sets under-cover in all but two simulations, and empirical coverage frequency may be as low as 59.5\%. Note that even for $\kappa = 1$ the success probabilities for the bottom categories are not well separated, which explains more pronounced under-coverage for the 7th category. In contrast, finite-sample methods (exactHolm and CP) cover the rank with the frequency no smaller than the desired level for all parametrizations and sample sizes. Both boot and bootStud confidence sets cover the 1st category with the frequency close to the nominal level, especially when success probabilities are equal ($\kappa=0$) and the sample size is small ($\tau = 0.5$). When the categories are better separated ($\kappa=1$) as in the dataset, then both finite-sample and bootstrap procedures cover the true rank with probability (close to) one.

Second, in contrast to our finding for Greater Melbourne, the finite-sample valid confidence sets are not uniformly tighter than the bootstrap confidences sets in simulations. In most parametrizations where boot and bootStud confidence sets are tighter, however, the difference in the size of the average confidence sets is below 0.1 and never exceeds 0.3, where the reference size of the confidence set covering the entire ranking is 6.0. 

Third, both CP and exactHolm methods produce much tighter confidence sets than both types of bootstrap for the 4th category when $\kappa = 1$, i.e., when success probabilities are equal to point estimates from Greater Melbourne respondents' sample.  Specifically, the average length of exactHolm confidence sets is less than one-third of the average length of boot and bootStud confidence sets when $\tau = 1$ and less than one-quarter of their average size when  $\tau = 2$. Furthermore, the average length of exactHolm confidence sets is more than 20\% lower than the average length of CP confidence sets in both instances. Despite the shorter average length, exactHolm confidence sets' empirical coverage frequency is above the desired level. Notice that both finite-sample valid confidence sets are of length zero in the majority of simulations with $\kappa = 1$ and $\tau = 2$ and contain only a single value 4. In contrast, most boot and bootStud confidence sets are of length one or above, meaning that they contain at least two values and are not as informative as finite-sample valid confidence sets.

\begin{table}[t]
\centering
\resizebox{\textwidth}{!}{ 
\begin{tabular}{>{\arraybackslash}p{3mm}>{\arraybackslash}p{3mm}
>{\arraybackslash}p{1mm}
>{\centering\arraybackslash}p{10mm}>{\centering\arraybackslash}p{10mm}>{\centering\arraybackslash}p{10mm}>{\centering\arraybackslash}p{10mm}>{\centering\arraybackslash}p{10mm}>{\centering\arraybackslash}p{10mm}
>{\arraybackslash}p{2mm}
>{\centering\arraybackslash}p{10mm}>{\centering\arraybackslash}p{10mm}>{\centering\arraybackslash}p{10mm}>{\centering\arraybackslash}p{10mm}>{\centering\arraybackslash}p{10mm}>{\centering\arraybackslash}p{10mm}
}
\toprule
&&&\multicolumn{6}{c}{\textbf{Length:}}&&\multicolumn{6}{c}{\textbf{Coverage Frequency:}}\\
\cmidrule(lr){4-9}\cmidrule(lr){11-16}
$\kappa$  & $\tau$ & & \makecell{exact \\Bonf}& \makecell{exact \\Holm}& CP & boot & \makecell{boot \\Stud} & naive &  & \makecell{exact \\Bonf} & \makecell{exact \\Holm} & CP & boot & \makecell{boot \\Stud}  & naive \\ 
 \midrule
  \multicolumn{14}{l}{\textbf{1st category:}}\\
\midrule
0 & 0.5 &  & 5.971 & 5.970 & 5.962 & 5.944 & 5.927 & 5.089 &  & 0.985 & 0.985 & 0.980 & 0.960 & 0.950 & 0.741 \\ 
   & 1 &  & 5.965 & 5.965 & 5.957 & 5.949 & 5.934 & 5.107 &  & 0.985 & 0.985 & 0.980 & 0.967 & 0.957 & 0.714 \\ 
   & 2 &  & 5.963 & 5.962 & 5.955 & 5.938 & 5.943 & 5.103 &  & 0.988 & 0.988 & 0.984 & 0.970 & 0.973 & 0.727 \\ 
  0.5 & 0.5 &  & 3.223 & 3.120 & 3.043 & 3.001 & 3.622 & 1.757 &  & 1.000 & 1.000 & 1.000 & 0.994 & 1.000 & 0.989 \\ 
   & 1 &  & 1.895 & 1.797 & 1.778 & 1.689 & 1.965 & 1.252 &  & 1.000 & 1.000 & 1.000 & 0.999 & 1.000 & 0.997 \\ 
   & 2 &  & 1.258 & 1.210 & 1.197 & 1.130 & 1.252 & 0.960 &  & 1.000 & 1.000 & 1.000 & 1.000 & 1.000 & 1.000 \\ 
  1 & 0.5 &  & 1.490 & 1.430 & 1.362 & 1.302 & 1.476 & 1.051 &  & 1.000 & 1.000 & 1.000 & 1.000 & 1.000 & 0.998 \\ 
   & 1 &  & 1.046 & 1.003 & 0.957 & 0.897 & 0.995 & 0.826 &  & 1.000 & 1.000 & 1.000 & 0.999 & 1.000 & 0.999 \\ 
   & 2 &  & 0.926 & 0.889 & 0.841 & 0.774 & 0.854 & 0.731 &  & 1.000 & 1.000 & 0.999 & 0.998 & 1.000 & 0.998 \\ 
    \midrule
  \multicolumn{14}{l}{\textbf{4th category:}}\\
\midrule
0 & 0.5 &  & 5.977 & 5.977 & 5.971 & 5.944 & 5.935 & 5.094 &  & 1.000 & 1.000 & 1.000 & 1.000 & 1.000 & 0.982 \\ 
   & 1 &  & 5.957 & 5.955 & 5.950 & 5.925 & 5.924 & 5.105 &  & 0.999 & 0.998 & 0.998 & 0.997 & 0.996 & 0.975 \\ 
   & 2 &  & 5.951 & 5.951 & 5.940 & 5.924 & 5.916 & 5.090 &  & 1.000 & 1.000 & 1.000 & 1.000 & 1.000 & 0.971 \\ 
  0.5 & 0.5 &  & 5.203 & 5.192 & 5.116 & 4.867 & 5.130 & 3.553 &  & 1.000 & 1.000 & 1.000 & 0.999 & 1.000 & 0.993 \\ 
   & 1 &  & 4.232 & 4.192 & 4.141 & 4.159 & 4.388 & 2.734 &  & 1.000 & 1.000 & 1.000 & 1.000 & 1.000 & 0.998 \\ 
   & 2 &  & 3.274 & 3.170 & 3.239 & 3.503 & 3.761 & 1.882 &  & 1.000 & 1.000 & 1.000 & 1.000 & 1.000 & 1.000 \\ 
  1 & 0.5 &  & 2.540 & 2.472 & 2.595 & 3.782 & 3.953 & 0.941 &  & 1.000 & 1.000 & 1.000 & 1.000 & 1.000 & 1.000 \\ 
   & 1 &  & 1.007 & 0.861 & 1.036 & 2.771 & 3.353 & 0.355 &  & 1.000 & 1.000 & 1.000 & 1.000 & 1.000 & 1.000 \\ 
   & 2 &  & 0.233 & 0.171 & 0.222 & 0.821 & 1.314 & 0.058 &  & 1.000 & 1.000 & 1.000 & 1.000 & 1.000 & 1.000 \\ 
    \midrule
  \multicolumn{14}{l}{\textbf{7th category:}}\\
\midrule
 0 & 0.5 &  & 5.968 & 5.966 & 5.959 & 5.914 & 5.910 & 5.079 &  & 0.985 & 0.985 & 0.983 & 0.984 & 0.982 & 0.672 \\ 
   & 1 &  & 5.970 & 5.968 & 5.961 & 5.946 & 5.937 & 5.112 &  & 0.985 & 0.985 & 0.981 & 0.985 & 0.984 & 0.672 \\ 
   & 2 &  & 5.967 & 5.967 & 5.959 & 5.951 & 5.948 & 5.185 &  & 0.985 & 0.985 & 0.980 & 0.986 & 0.985 & 0.728 \\ 
  0.5 & 0.5 &  & 4.315 & 4.288 & 4.267 & 3.994 & 4.253 & 3.036 &  & 0.995 & 0.995 & 0.995 & 1.000 & 1.000 & 0.908 \\ 
   & 1 &  & 3.393 & 3.345 & 3.366 & 3.354 & 3.550 & 2.537 &  & 0.996 & 0.994 & 0.996 & 1.000 & 1.000 & 0.943 \\ 
   & 2 &  & 2.802 & 2.748 & 2.799 & 2.865 & 3.018 & 2.214 &  & 0.998 & 0.998 & 0.999 & 1.000 & 1.000 & 0.964 \\ 
  1 & 0.5 &  & 2.331 & 2.331 & 2.393 & 2.777 & 5.097 & 1.220 &  & 1.000 & 1.000 & 1.000 & 1.000 & 1.000 & 0.595 \\ 
   & 1 &  & 1.892 & 1.879 & 1.930 & 2.247 & 3.926 & 1.237 &  & 1.000 & 1.000 & 1.000 & 1.000 & 1.000 & 0.856 \\ 
   & 2 &  & 1.525 & 1.492 & 1.630 & 2.003 & 2.277 & 1.034 &  & 1.000 & 1.000 & 1.000 & 1.000 & 1.000 & 0.986 \\ 
\bottomrule
\end{tabular}
}
\caption{ Average lengths and empirical coverage frequencies from 1000 Monte Carlo samples for the 95\% marginal confidence sets for the rank of the 1st (\textbf{top panel}), 4th (\textbf{middle panel}) and 7th (\textbf{bottom panel}) categories. } 
\label{tab: TabLengthCoverage}
\end{table}

Fourth, when the success probabilities are equal to point estimates from the respondents' sample of Greater Melbourne ($\kappa =1$), CP and exactHolm produce much tighter confidence sets for the 7th category than bootStud. For $\tau = 0.5$ and $\tau = 1$ the average length of CP and exactHolm confidence sets is more than two times smaller than the average length of bootStud confidence set. The difference becomes less pronounced with the increase in sample size, but the average length of the exactHolm confidence set is still more than 50\% smaller than the average length of bootStud confidence set for $\tau = 2$. Furthermore, the exactHolm confidence set is substantially shorter than the boot and CP confidence sets for all values of $\tau$. Tighter exactHolm confidence sets for the 7th category still provide empirical coverage frequency above the desired level.

Table~\ref{tab: TabLengthCoverage} also highlights common features of the inference procedures. For equal success probabilities ($\kappa=0$) the set of ranks $R_j$ is $[1,7]$ for all categories $j$, and the average length of all confidence sets barely decreases for larger sample sizes $\tau$. When success probabilities differ ($\kappa>0$) larger sample size $\tau$ means differences in success probabilities are easier to detect, and all confidence sets for the rank decrease in length. Similarly, the average length of all confidence sets except for bootStud for 7th category decreases as we increase $\kappa$, which means differences in success probabilities are larger and thus, again, easier to detect. For bootStud at the bottom of the ranking, the effect of better separation for $\kappa>0$ is mitigated by decreasing success probabilities leading to more frequent division by zero events in bootstrap test statistics and thus larger critical values. 

Finally, in addition to the five confidence sets we used in Section~\ref{sec: emp} we include exactBonf in all simulations. Table~\ref{tab: TabLengthCoverage} shows that, as expected, exactBonf confidence sets are uniformly wider than exactHolm confidence sets, but the difference in the average length is not substantial.

\subsection{Erratic Coverage}
\label{sec: erratic}

\cite{Brown:2001p5787} found that coverage frequencies of some confidence intervals for binomial proportions may vary in highly non-monotonic ways with the sample size and the success probability. Furthermore, they found that coverage frequencies may be far below the desired level, especially for small sample sizes and/or small success probabilities. Motivated by this ``erratic'' behavior of coverage in the binomial case, in this subsection, we compare our confidence set (exactBonf), which is valid in finite samples, with the bootstrap confidence sets (boot and bootStud), which are justified by asymptotic validity. In particular, we are interested in their coverage properties in small samples and/or scenarios with small success probabilities. To this end we consider three categories and set the vector of success probabilities as $\theta = (\pi, \pi, 1-2\pi)$, where $\pi$ is varied between $1/100$ and $1/3$. The sample size is varied between $10$ and $100$.

For the different values of $\pi$, Figure~\ref{fig: coverage erratic diff} shows the frequencies of the confidence set $C_{{\rm symm},n}(1-\alpha,I)$ simultaneously covering all differences involving the first category, $\Delta_I$, where $I=J^{\rm two-sided}$ with $J_0=\{1\}$. The coverage frequencies are plotted as functions of the sample size $n$ for the bootstrap with (panel~(a)) and without (panel~(b)) studentization. Both bootstrap approaches lead to simultaneous confidence sets for the differences that considerably under-cover for small sample sizes and or small probabilities $\pi$, analogously to the findings in \cite{Brown:2001p5787}. The coverage probability for $\pi=1/100$ can be even lower than $0.4$ when sample sizes are small ($n\leq 20$).

Figure~\ref{fig: coverage erratic rank} shows the coverage frequencies of the resulting confidence sets for the rank based on the two bootstrap approaches (panels~(a) and (b)) and, for comparison, also the coverage frequencies of the finite sample method (panel~(c)). Perhaps somewhat surprisingly, the coverage frequencies for the bootstrap methods are above $1-\alpha$ in all scenarios, even when success probabilities and/or sample sizes are small. Hence, under-coverage of the differences observed in Figure~\ref{fig: coverage erratic diff} does not lead to under-coverage of the rank. The reason for this phenomenon is that correct coverage of the rank only requires that the confidence sets for the differences do not lead to incorrect claims about the signs of the differences. While not necessarily covering the true values of the differences, the bootstrap confidence sets for the differences lead to the correct determination of their signs and thus to coverage of the rank. As expected, the finite-sample method covers the rank with the desired probability even for small samples sizes and/or small success probabilities.

\begin{figure}[!t]
  \centering
  \includegraphics[width=\textwidth]{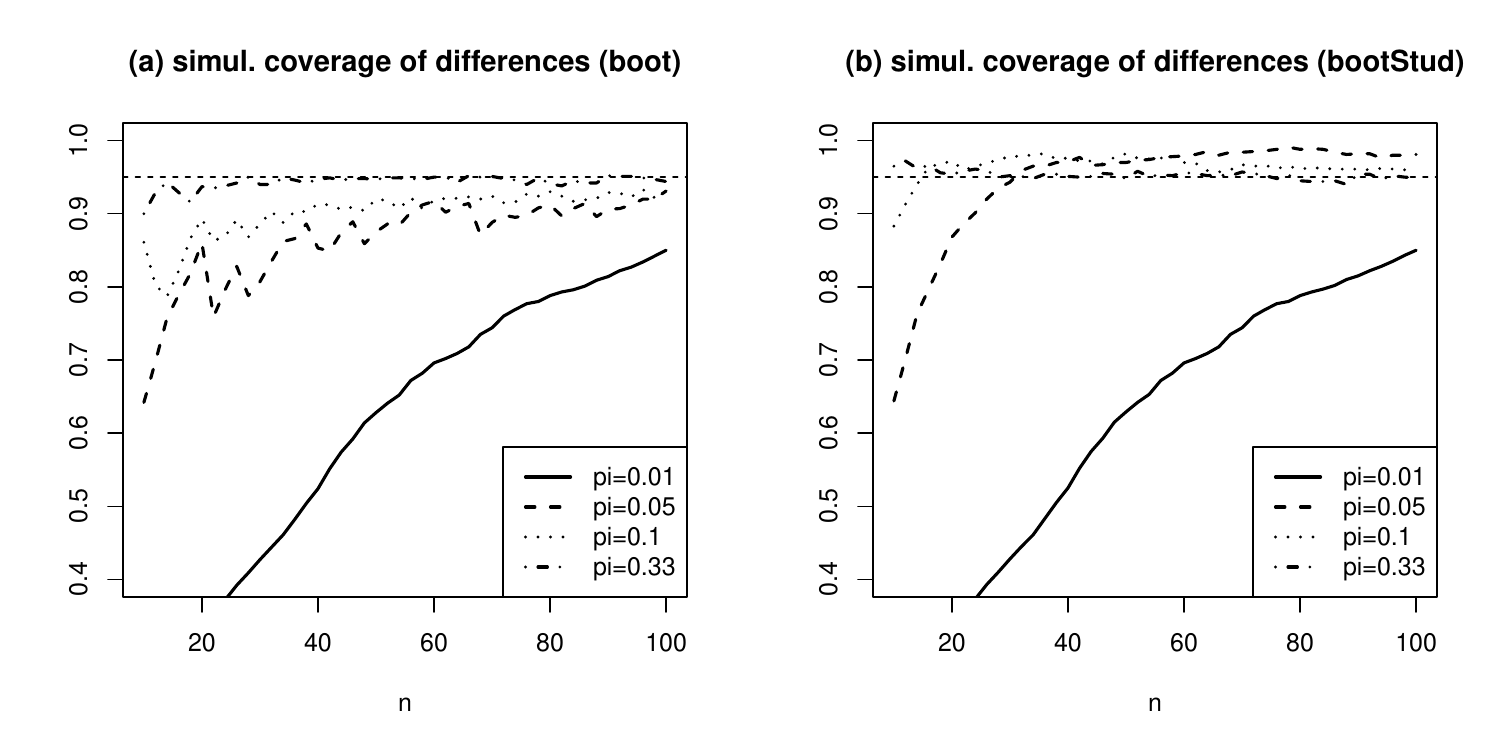}
  \caption{Frequencies of simultaneously covering all differences involving the first category, i.e., the frequency of $\Delta_I\in C_{{\rm symm},n}(1-\alpha,I)$. The horizontal dashed line marks the desired coverage level $1-\alpha$.} \label{fig: coverage erratic diff}
\end{figure}

\begin{figure}[!t]
  \centering
  \includegraphics[width=\textwidth]{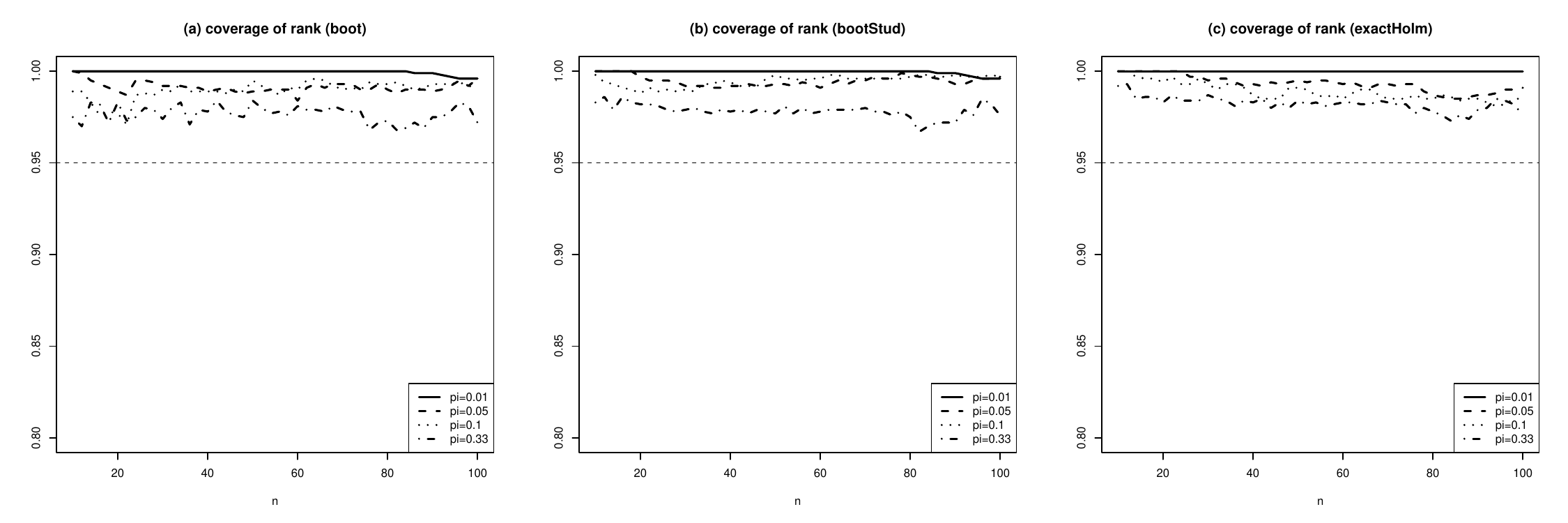}
  \caption{Coverage frequencies for the rank of the first category. The horizontal dashed line marks the desired coverage level $1-\alpha$.} \label{fig: coverage erratic rank}
\end{figure}

\subsection{Larger Number of Categories}
\label{sec: sim largep}

In this subsection, we consider a simulation design in which the success probabilities are all equal, i.e. $\theta_j=1/p$ for all $j=1,\ldots,p$, and we increase $p$ from $5$ to $50$. For the different values of $p$, Figures~\ref{fig: coverage plarge} and \ref{fig: length plarge} show the coverage frequencies and lengths of the different confidence sets as functions of the sample size $n$.

First of all, the finite-sample methods exactBonf, exactHolm, and CP cover the rank with coverage frequencies close to one in all scenarios. As expected the naive bootstrap fails to cover the true rank with the desired probability. Its coverage may be considerably below $0.4$ and even approaches zero when there are many categories. Interestingly, however, the bootstrap method based on the studentized statistic also considerably under-covers when there are (moderately) many categories. For instance, with $p=20$ categories, its coverage frequency can be well below $0.8$ for small sample sizes. For more categories ($p=50$), the under-coverage occurs up to larger sample sizes.

In terms of length, the exactBonf, exactHolm, CP, and boot perform similarly. As expected bootStud and naive lead to shorter confidence sets in those scearios in which they under-cover.

\begin{figure}[!t]
  \centering
  \includegraphics[width=0.9\textwidth]{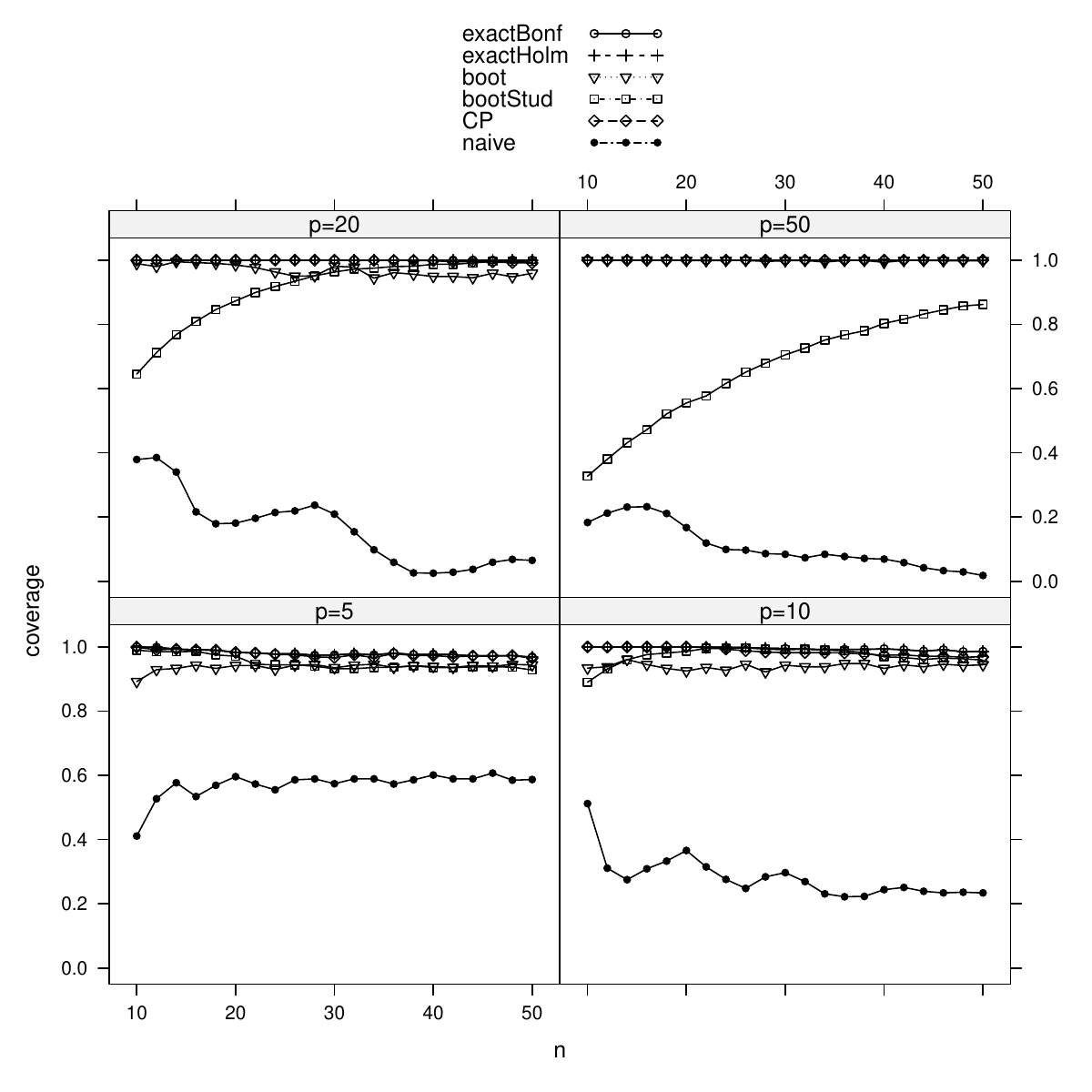}
  \caption{Coverage frequencies of confidence sets for different $n$ and $p$.} \label{fig: coverage plarge}
\end{figure}

\begin{figure}[!t]
  \centering
  \includegraphics[width=0.9\textwidth]{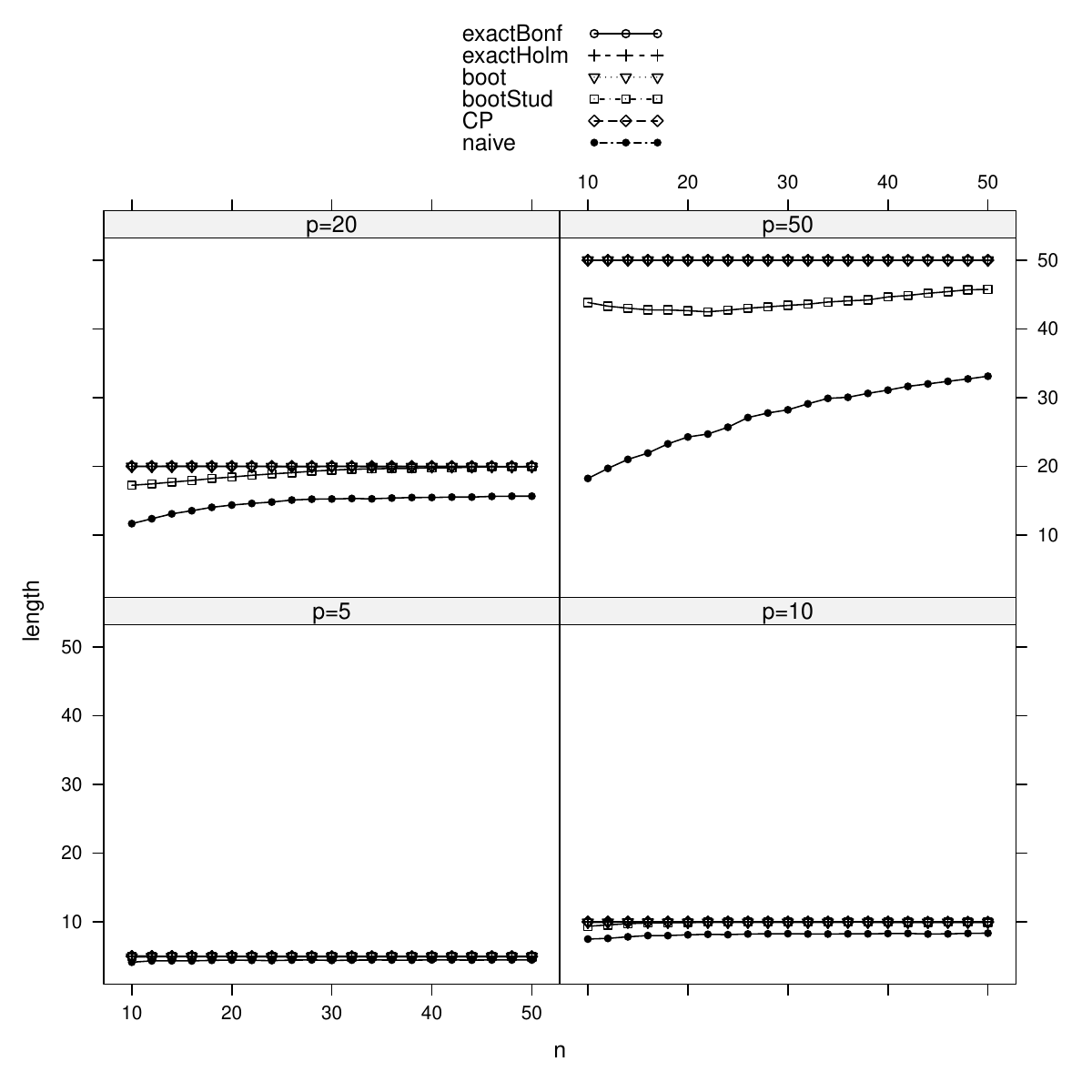}
  \caption{Length of confidence sets for different $n$ and $p$.} \label{fig: length plarge}
\end{figure}

\clearpage
\newpage
\appendix

\section{Asymptotic Validity of the Bootstrap}
\label{app: bootstrap}

For the arguments in this section, we slightly change notation by indexing population quantities by $P$, the underlying probability mechanism that specifies the multinomial sampling probabilities. For instance, let $\theta\equiv \theta(P)\equiv (\theta_1(P),\ldots,\theta_p(P))'$ denote the probabilities of a particular $P$ and similarly
\begin{equation*}
	\Delta_I(P) \equiv (\Delta_{j,k}(P) : (j,k) \in I)~,
\end{equation*}
where $$\Delta_{j,k}(P) \equiv \theta_j (P) - \theta_k (P)$$ and $I \subseteq J^2$. As in the main text let
\begin{equation}\label{equation:hatsigmajk}
\hat \sigma^2_{j,k} \equiv  \hat \theta_j ( 1-  \hat \theta_j ) + \hat \theta_k (1-  \hat \theta_k ) + 2 \hat \theta_j \hat \theta_k
\end{equation}
define the cumulative distribution functions 
\begin{eqnarray}
	L_{{\rm lower}, n}(x,I, P) &\equiv& P\left \{ \max_{(j,k) \in I} \frac{\hat{\theta}_j - \hat{\theta}_k - \Delta_{j,k}(P)}{\hat{\sigma}_{j,k} / \sqrt{n} } \leq x \right \} ~, \label{eq:Lnlower} \\
	L_{{\rm upper}, n}(x,I, P) &\equiv& P\left \{ \max_{(j,k) \in I} \frac{\Delta_{j,k}(P) - (\hat{\theta}_j - \hat{\theta}_k)}{\hat{\sigma}_{j,k}  / \sqrt{n} } \leq x \right \} ~, \label{eq:Lnupper} \\
	L_{{\rm symm}, n}(x,I, P) &\equiv& P\left \{ \max_{(j,k) \in I} \frac{|\hat{\theta}_j - \hat{\theta}_k - \Delta_{j,k}(P)|}{\hat{\sigma}_{j,k} / \sqrt{n} } \leq x \right \} ~. \label{eq:Lnsymm}
\end{eqnarray}
Let $\hat P_n$ be an estimate of $P$, where $\hat P_n$  specifies the empirical frequencies $\hat \theta = ( \hat \theta_ 1 , \ldots , \hat \theta_p )~.$ Then the bootstrap quantiles can be written as
$$c_{l,n}(1-\alpha,I) = L_{l, n}^{-1}(1-\alpha,I, \hat{P}_n) $$
for $l\in \{\rm lower,upper,symm \}$. Here, it is understood that, for a cumulative distribution function $F(x)$ on the real line, the quantity $F^{-1}(1  - \alpha)$ is defined to be $\inf\{x \in \mathbf R: F(x) \geq 1 - \alpha\}$. The bootstrap simply replaces the unknown frequencies $\theta$ with its empirical counterpart $\hat \theta$, i.e. $\hat{\theta} = \theta(\hat{P}_n)$.

Consider the rectangular confidence set for the vector of differences $\Delta_I(P)$ defined by
\begin{equation*}
	C_{l, n}(1 - \alpha,I) \equiv \prod_{(j,k) \in I} C_{l, n,j,k}(1 - \alpha,I)
\end{equation*}
where $C_{l, n,j,k}(1 - \alpha,I)$ could be defined in various ways:
\begin{align}
	C_{{\rm lower}, n,j,k}(1 - \alpha,I) &\equiv \Biggr [ \hat{\theta}_j - \hat{\theta}_k - c_{{\rm lower}, n}(1- \alpha,I) \frac{ \hat{\sigma}_{j,k}}{\sqrt{n}}, \infty \Biggr ). \label{eq:Cnlower}\\
	C_{{\rm upper}, n,j,k}(1 - \alpha, I) &\equiv \Biggr (-\infty, \hat{\theta}_j - \hat{\theta}_k + c_{{\rm upper}, n}(1- \alpha,I) \frac{ \hat{\sigma}_{j,k}}{\sqrt{n}} \Biggr ], \label{eq:Cnupper}\\
	C_{{\rm symm}, n,j,k}(1 - \alpha, I) &\equiv \Biggr [\hat{\theta}_j - \hat{\theta}_k \pm c_{{\rm symm}, n}(1- \alpha,I) \frac{ \hat{\sigma}_{j,k}}{\sqrt{n}} \Biggr ], \label{eq:Cnsymm} \\
	C_{{\rm equi}, n,j,k}(1 - \alpha, I) &\equiv C_{{\rm lower}, n}\left (1 - \frac{\alpha}{2}, I \right ) \bigcap C_{{\rm upper}, n}\left (1 - \frac{\alpha}{2}, I \right ).\label{eq:Cnequi}
\end{align}
The following lemma shows that these bootstrap confidence sets are asymptotically valid in the sense that they cover the true vector of differences with probability approaching $1-\alpha$:

\begin{lemma}\label{lem: bootstrap validity for CS for differences}
	For any $l\in \{\rm lower,upper,symm \}$ and any $I\subset J^2$,
	$$\lim_{n\to\infty} P\left\{\Delta_I(P) \in C_{l,n}(1-\alpha,I)\right\} \geq 1-\alpha$$ with equality if $\theta_j(P) > 0$ for all $j \in J$.
\end{lemma}

\begin{proof}
	We begin by considering the case where $\theta_j(P) > 0$ for all $j \in J$.  First, consider the joint behavior of $( \sqrt{n} ( \hat \theta_{j} - \hat \theta_k )/ \hat \sigma_{j,k} )$ for
	all ${p \choose 2}$ distinct pairs $(j,k)$ with $j \ne k$.  Toward this end,  let  
	$J_n ( P)$ denote the joint distribution of  $\sqrt{n} ( \hat \theta_1 - \theta_ 1(P) , \ldots , \hat \theta_p - \theta_p(P) )$.
	By the multivariate Central Limit Theorem, $J_n ( P)$ converges in distribution to $J(P)$,  the multivariate normal distribution with mean 0 and covariance matrix $\Sigma = \Sigma (P)$, where $\Sigma$ has $(j,k)$ entry
	$\theta_j(P) ( 1- \theta_j(P) )$ if $j = k$ and $ - \theta_j(P) \theta_k(P)$ if $j \ne k$.
	Moreover,  in a triangular array setup, if $P_n$ is a sequence of multinomial probabilities with 
	$\theta (P_n ) \to \theta (P)$, then $J_n ( P_n )$ converges in distribution  to $J(P)$. To see why,
	apply the Cram\'er-Wold device and the Lindeberg CLT.
	Since, $\theta ( \hat P_n ) \to \theta (P)$ almost surely (componentwise, by the Strong Law of Large Numbers), it follows that 
	$$\rho \left  ( J_n ( \hat P_n ) , J_n ( P) \right ) \to 0~~~{\rm almost~surely}~,$$
	where $\rho$ is any metric metrizing weak convergence in ${\bf R}^p$.
	Next, let $J_n' (P)$ denote the joint  distribution of $$\sqrt{n} [  ( \hat \theta_j - \hat{\theta}_k ) - (\theta_j(P) - \theta_k(P) ) ]$$ for all $j < k$.
	So $J'_n (P)$ is a distribution on ${\bf R}^{p'}$, where $p' = {p \choose 2}$.  (The pairs can be
	ordered in any fashion, but  for the sake of argument, they are ordered as $(1,2), \ldots (1, p)$ followed by $(2,3) , \ldots , (2,p)$, etc.)
	By the Continuous Mapping Theorem,  $J'_n (P)$ converges in distribution to $J' (P)$,  the multivariate normal distribution with mean 0 and covariance matrix $\Sigma'  = \Sigma' (P)$.  Note $\Sigma'$ can easily be obtained from $\Sigma$, but its exact form is not actually required.  Again, this convergence is
	locally uniform in the sense that $J'_n (P_n )$ converges in distribution to $J' (P)$ whenever $\theta ( P_n ) \to \theta (P)$.  Since $\theta ( \hat P_n ) \to \theta (P)$ almost surely, we have
	$$\rho' \left  ( J'_n ( \hat P_n ) , J'_n ( P) \right ) \to 0~~~{\rm almost~surely}~,$$
	where $\rho'$ metrizes weak convergence on ${\bf R}^{p'}$.
	Finally, we can consider the joint distribution of studentized differences; to this end,
	let $J_n^* (P)$ denote the joint distribution of  the $p'$ variables
	 \begin{equation}\label{equation:joint*}
	  \frac { \sqrt{n} [  ( \hat \theta_j -  \hat \theta_k ) - (\theta_j(P) - \theta_k(P) ) ]}{ \hat \sigma_{j,k} }~,
	  \end{equation}
	 where $\hat \sigma^2_{j,k}$ is given in (\ref{equation:hatsigmajk}).
	Under $P$, $\theta ( \hat P_n ) \to \theta (P)$ almost surely, and so $\hat \sigma^2_{j,k}$ converges almost surely to $\sigma^2_{j,k} (P)$ given by
	$$\sigma_{j,k}^2  (P) = \theta_j(P) (1- \theta_j(P) ) + \theta_k(P) ( 1- \theta_k(P) ) +2 \theta_j(P) \theta_k(P)~.$$
	by a multivariate Slutsky Theorem (or the Continuous Mapping Theorem), $J_n^* ( P)$ converges in distribution to $J^* (P)$, the multivariate normal distribution with mean 0 and covariance matrix $\Sigma^* =\Sigma^* (P)$,
	where $\Sigma^*$ is easily obtained from $\Sigma'$ (as $\Sigma^*$ is the correlation matrix corresponding to the covariance matrix $\Sigma'$). 
	Under $P_n$ with $\theta (P_n ) \to \theta (P)$, it also follows that $\hat \sigma_{j,k}$ converges almost surely to $ \sigma_{j,k}(P)$.  To see why, first show $\hat \theta_j$ converges to $\theta_j$ with probability one under $P_n$;  since $\hat \theta_j$ can be viewed an average of bounded i.i.d.\ variables, this convergence follows easily by the well-known 4th moment argument and the Borel-Cantelli Lemma.  Hence, under $P_n$, we also have $J^* ( P_n ) $ converges in distribution to $J^* (P)$, and then for the same reasons as for $J_n$ and $J'_n$, we also have
	$$\rho' \left  ( J^*_n ( \hat P_n ) , J^*_n ( P) \right ) \to 0~~~{\rm almost~surely}~$$
	and
	$$\rho' \left  ( J^*_n ( \hat P_n ) , J^* ( P) \right ) \to 0~~~{\rm almost~surely}.~$$
	All of these results carry over if we consider a subset $I \subset J^2$, by the Continuous Mapping Theorem.  For example, if $J_n^* (I, P)$ refers to the joint distribution of the variables (\ref{equation:joint*}), but only for $(j,k) \in I$, then it follows that
	$$\rho' \left  ( J^*_n ( I,  \hat P_n ) , J^*_n (I,  P) \right ) \to 0~~~{\rm almost~surely}~.$$

	Bootstrap consistency of the distributions \eqref{eq:Lnlower}--\eqref{eq:Lnsymm} now follows
	by the Continuous Mapping Theorem.  Indeed, for any $P_n$ with $\theta ( P_n ) \to \theta (P)$,
	$$L_{{\rm lower},n} (x, I, P_n ) \to L_{{\rm lower},n} (x, I, P )~~{\rm for~all~}x,$$
	where $L_{{\rm lower}, n} ( \cdot , I, P ) )$ is the distribution of $\max_{ (j,k) \in I } Z_{j,k}$
	and  $Z = ( Z_{j,k} : (j,k) \in I  ) $ is multivariate normal  with distribution $J^* (I, P)$.  Note that this distribution
	is continuous everywhere and strictly increasing.  
	Hence, since $\theta ( \hat P_n ) \to \theta (P)$ almost surely, we also have, for all $x$,
	$$L_{{\rm lower},n} (x, I,  \hat P_n ) \to L_{{\rm lower},n} (x, I, P )~~{\rm almost~surely}~.$$
	It also follows that bootstrap quantiles are consistent in the sense that
	$$L_{{\rm lower},n}^{-1} (1- \alpha , I, \hat P_n ) \to L_{{\rm lower},n}^{-1} (1- \alpha , I, P )~~{\rm almost~surely}~.$$
	By Slutsky,
	$$P\left \{ \max_{(j,k) \in I} \frac{\hat{\theta}_j - \hat{\theta}_k - \Delta_{j,k}(P)}{\hat{\sigma}_{j,k} / \sqrt{n} } \leq  L_{{\rm lower},n}^{-1} (1- \alpha , I, \hat P_n )   \right \}  \to 1- \alpha~.$$
	By ``inverting" this probability statement,  we can conclude the intervals 
	$$ \left [ \hat \theta_j - \hat \theta_k - \frac{ \hat \sigma_{j,k} }{\sqrt{n}}  L_{{\rm lower},n}^{-1} (1- \alpha , I, \hat P_n ) ,~ \infty \right )$$
	jointly cover the true $\theta_j(P) - \theta_k(P)$ with asymptotic probability $1- \alpha$.
	The arguments for upper and two-sided confidence bounds are analogous.

	The above argument maintained the assumption that $\theta_j(P) > 0$ for all $j \in J$.  We now consider the case where that need not be true.  To this end, first suppose that $\theta_j(P) = \theta_k(P) = 0$ for all $(j,k) \in I$.  In this case, $\hat \theta_j = \hat \theta_k = \hat \sigma_{j,k} = 0$ for all $(j,k) \in I$.  But then, bootstrap samples from $\hat P_n$ also satisfy $\hat \theta_j^*  = \hat \theta_k^*  = \hat \sigma_{j,k}^* = 0$ for all $(j,k) \in I$. Hence, our convention that $0/0 = 0$ and $c/0 = \text{sign}(c)\infty$ for $c \neq 0$, implies that $L^{-1}_{{\rm lower}, n}(1 - \alpha, I, \hat P_n) = 0$ w.p.1.  It follows that $$P\{ \Delta_I(P) \in C_{{\rm lower}, n}(1 - \alpha,I) \} = 1~.$$  Now suppose that $\theta_j(P) > 0$ or $\theta_k(P) > 0$ for some $(j,k) \in I$.  The same argument implies that $$P\{ \theta_j(P) - \theta_k(P) \in C_{{\rm lower}, n,j,k}(1 - \alpha,I) \} = 1$$ for any $(j,k) \in I$ with $\theta_j(P) = \theta_k(P) = 0$.  To complete the proof, it suffices to show that 
	\begin{equation} \label{eq:enough}
	\liminf_{n \rightarrow \infty} P\{\theta_j(P) - \theta_k(P) \in C_{{\rm lower}, n,j,k}(1 - \alpha,I) \text{ for all } (j,k) \in I^*\} \geq 1 - \alpha,
	\end{equation}
	where $$I^* = \{(j,k) \in I : \theta_j(P) > 0 \text{ or } \theta_k(P) > 0 \}.$$  Since $$L^{-1}_{{\rm lower}, n}(1 - \alpha, I, \hat P_n) \geq L^{-1}_{{\rm lower}, n}(1 - \alpha, I^*, \hat P_n),$$ the desired convergence in \eqref{eq:enough} can be established simply by arguing as in the first part of the theorem.
	
\end{proof}

\section{Proofs of the Main Results}

\begin{proof}[Proof of Theorem~\ref{thm: coverage}]
	We prove the theorem for the two-sided case ($I=J^{\rm two-sided}$), but the derivations are very similar for the one-sided cases. Define
	\begin{align*}
		S^- &\equiv \bigcup_{j\in J_0} S_j^- \qquad \text{with} \qquad S_j^-\equiv \left\{(j,k)\in I\colon j\neq k \text{ and } \theta_j\leq \theta_k \right\}\\
		S^+ &\equiv \bigcup_{j\in J_0} S_j^+ \qquad \text{with} \qquad S_j^+\equiv \left\{(j,k)\in I\colon j\neq k \text{ and } \theta_j\geq \theta_k \right\}
	\end{align*}
	and
	\begin{align*}
		R^{-} &\equiv \bigcup_{j\in J_0} R^{-}_j \qquad \text{with} \qquad R^{-}_j\equiv \left\{(j,k)\in I\colon j\neq k, \text{ reject } H_{k,j}, \text{ and claim } \theta_j<\theta_k \right\}\\
		R^{+} &\equiv \bigcup_{j\in J_0} R^{+}_j \qquad \text{with} \qquad R^{+}_j\equiv \left\{(j,k)\in I\colon j\neq k, \text{ reject } H_{j,k}, \text{ and claim } \theta_j>\theta_k \right\}
	\end{align*}
	Suppose $S^-\cap R^{+}=\emptyset$ and $S^+\cap R^{-}=\emptyset$. Then:
	\begin{align*}
		&\forall j \in J_0\colon S_j^-\cap R^{+}_j=\emptyset \text{ and }S_j^+\cap R^{-}_j=\emptyset\\
		\Rightarrow\qquad &\forall j \in J_0\colon \theta_j>\theta_k \; \forall (j,k)\in R^{+}_j \text{ and } \theta_j<\theta_k\; \forall (j,k)\in R^{-}_j\\
		\Rightarrow\qquad &\forall j \in J_0\colon \theta_j>\theta_k \; \forall k\in \text{Rej}^{+}_j \text{ and } \theta_j<\theta_k\; \forall k\in \text{Rej}^{-}_j\\
		\Rightarrow\qquad &\forall j \in J_0\colon r_j \leq p-|\text{Rej}^{+}_j| \text{ and } r_j \geq 1+|\text{Rej}^{-}_j|
	\end{align*}
	The third implication uses the fact that the number of pairs $(j,k)$ in $R^{+}_j$ (or $R^{-}_j$) is equal to the number of of $k$ in $\text{Rej}^{+}_j$ (or $\text{Rej}^{-}_j$). Therefore,
	$$P\left\{r_{j} \in R_{n,j}\;\forall j\in J_0\right\} \geq P\left\{ S^-\cap R^{+}=\emptyset \text{ and } S^+\cap R^{-}=\emptyset\right\} = 1- FWER_I $$
	and the desired result follows from \eqref{eq: FWER control}.
\end{proof}

\begin{proof}[Proof of Theorem~\ref{thm: size control}]
	We first show that the distribution of $X_j$ given $S_{j,k}=s$ is binomial based on $s$ trials and success probability $\theta_j/(\theta_j+\theta_k)$. To see this note that
	$$P \{ X_j = x | S_{j,k} = s \}  \propto P  \{ X_j = x \text{ and } X_k = s-x \} = P \{ X_j = x  \} P \{ X_k = s-x | X_j = x \},$$
	where the symbol $\propto$ means there is a constant out in front that can depend on $s$ and $\theta_j , \theta_k$ (which we'll see depends on $\theta_j$ and $\theta_k$ through $\theta_j  / ( \theta_j + \theta_k )$). Continuing,
	$$P \{ X_j = x | S_{j,k} = s \}  \propto  {n \choose x} \theta_j^x ( 1 - \theta_j )^{n-x}  { {n-x} \choose {s-x}}
	\left ( \frac{\theta_k}{ 1- \theta_k} \right )^{s-x} \left (
	 \frac{ 1- \theta_j - \theta_k}{1 - \theta_j} \right )^{n-s}$$
	$$ \propto  { s \choose x } (\theta_j / \theta_k )^x \propto {s \choose x } \left ( \frac{\theta_j}{\theta_j + \theta_k} \right )^x \left ( \frac{\theta_k}{\theta_j + \theta_k} \right )^{s-x}~,
	$$
	which is the binomial family of distributions with univariate parameter $\theta_{j,k} = \theta_j / ( \theta_j + \theta_k )$.
	Clearly, a test of $H_{j,k}$ is equivalent to a the hypothesis specifying $\theta_{j,k} \le 1/2$.
	As is well known, this family of distributions has monotone likelihood ratio.  Therefore,
	by Corollary~3.4.1 of \cite{Lehmann:2022aa}  and conditional on $S_{j,k}$, there exists a UMP level $\beta_{j,k}$ test for testing $H_{j,k}$ given by  \eqref{eq: def of test} with constants $\gamma(s)$ and $C(s)$ determined by
	$$E_{\theta_k} [\phi(X_j,S_{j,k})| S_{j,k}=s] = \beta_{j,k}\quad \forall s. $$
	Here, $E_{\theta_k}[\cdot]$ refers to the expectation under which $\theta_j=\theta_k$. It is easy to see that the equation determining the constants can be written as in \eqref{eq: def of constants}, which shows that the test has exact rejection probability (conditional on $S_{j,k}$) equal to $\beta_{j,k}$.
	 By monotone likelihood ratio (and still conditional
	on $S_{j,k}$), the (conditional) power function of the test is nondecreasing, and so the conditional rejection probability
	is bounded above by $\beta_{j,k}$ for all $\theta_j$ and $\theta_k$ satisfying $\theta_{j,k} \le 1/2$
	
	Thus far, we have shown that $\phi (X_j , S_{j,k} )$ is the UMP level $\beta_{j,k}$ test, conditional on $S_{j,k}$.
	In order to argue that it is UMPU level $\beta_{j,k} $ among all level $\beta_{j,k}$ unbiased tests  (unconditionally),
	consider the boundary of the parameter space $\omega_{j,k} = \{ ( \theta_1 , \ldots , \theta_p );~\theta_j = \theta_k \}$.
	The family of distributions of $(X_1 , \ldots , X_p )$ still is multinomial, but now  $T = (  S_{j,k}, ~X_i , i \ne j, i \ne k )$
	is complete and sufficient for $\omega_{j,k}$.  Hence, any test, say $\psi$, that is similar on $\omega_{j,k}$, i.e., it satisfies
	that the rejection probability is equal to $\beta_{j,k}$ for  all $\theta \in \omega_{j,k}$ or $E_{\theta} ( \psi ) = \beta_{j,k}$ for all $\theta \in \omega_{j,k}$, must satisfy
	that it is conditionally level $\beta_{j,k}$ given $T$ or $E_{\theta} ( \psi | T) = \beta_{j,k}$ for all $\theta \in \omega_{j,k}$.
	In other words, all similar tests have Neyman structure, by Theorem 4.3.2 in \cite{Lehmann:2022aa}.  Therefore,
	the optimal  unconditional level $\beta_{j,k}$ test must  be obtained by finding the optimal conditional level $\beta_{j,k}$,
	conditional on $T$.  Finally, note that specifying the  conditional distribution of the data $(X_1 , \ldots , X_p)$ given $T$
	is equivalent to specifying the conditional distribution of $X_j$ given $S_{j,k}$, since conditionally all other $X_i$ 
	with $i \ne j$ and $i \ne k$ are now fixed.  But, we have found the optimal conditional test above based on the conditional distribution of $X_j$ given $S_{j,k}$. 
	 (Note that we could have argued by writing the family of distributions in the canonical 
	multiparameter exponential form discussed in Section 4.4 of \cite{Lehmann:2022aa}, but the notation becomes
	messy, stemming from the fact that the rank of the multinomial family is $p-1$ and not $p$, and consequently the argument gets obscured.)
\end{proof}

\begin{proof}[Proof of Theorem~\ref{thm: bootstrap validity}]
	First, by an argument analogous to the one in Theorem~3.3 in \cite{Mogstad:2024aa},
	$$P\left\{r_{j} \in R_{n,j}^{\rm boot}\;\forall j\in J_0\right\} \geq P\left\{\Delta_I \in C_{{\rm lower},n}(1-\alpha,I)\right\}, $$
	where $I$ is equal to one of the sets $J^{\rm lower}$, $J^{\rm upper}$, or $J^{\rm two-sided}$, which was used to construct $R_n^{\rm boot}$. The desired result therefore follows from Lemma~\ref{lem: bootstrap validity for CS for differences}.
\end{proof}

\section{Supporting Results for the Empirical Application}

\begin{figure}[H]
    \centering
    \ContinuedFloat*
    \includegraphics[width=\textwidth]{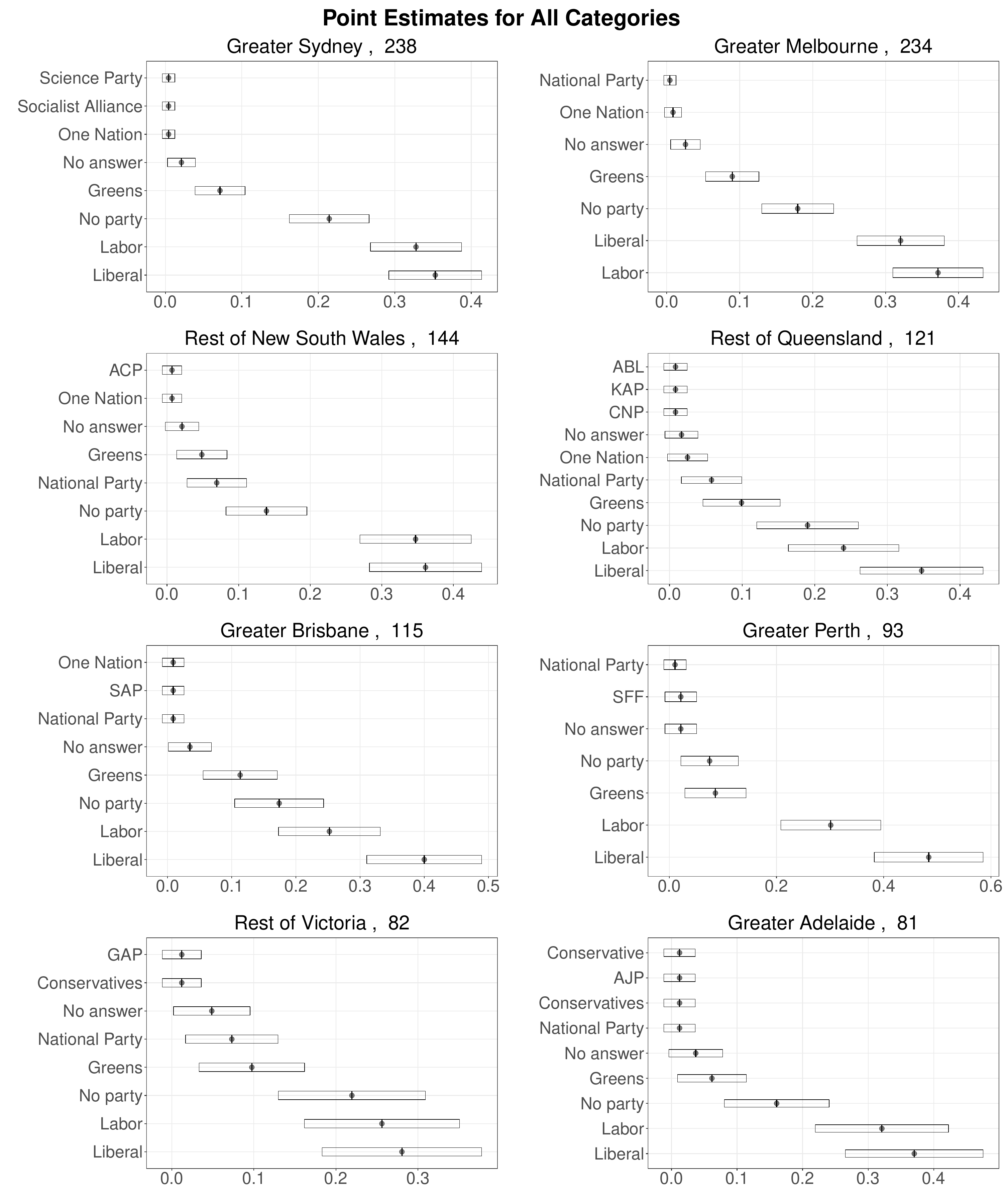}
    \caption{ Point estimates of categories support shares in fifteen Australian territories from AES2019 and $\pm 1.96$se. }
    \label{fig: Data1}
\end{figure}

\begin{figure}[H]
    \centering
    \ContinuedFloat
    \includegraphics[width=\textwidth]{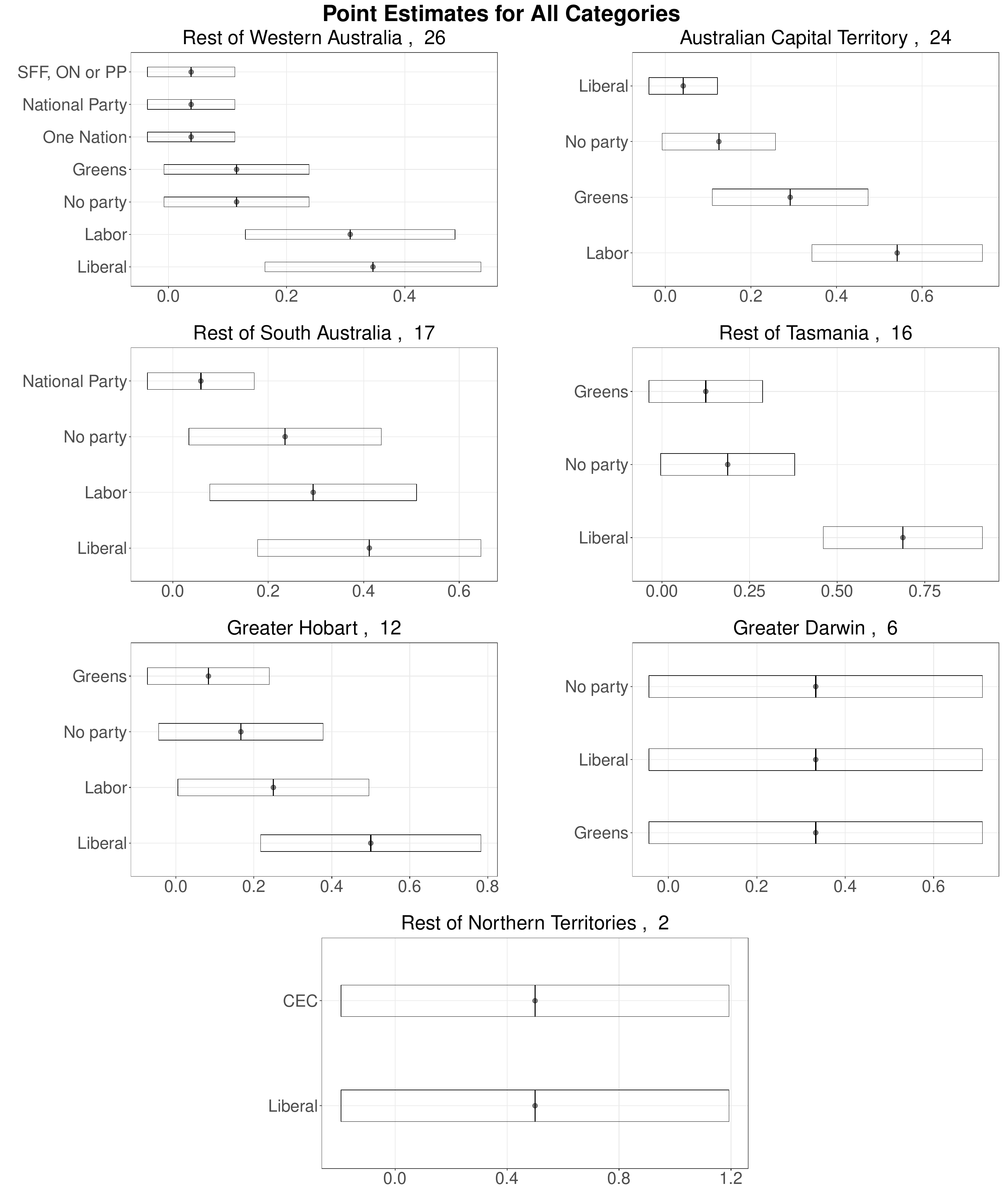}
    \caption{Point estimates of categories support shares in fifteen Australian territories from AES2019 and $\pm 1.96$se.}
    \label{fig: Data2}
\end{figure}

\begin{figure}[H]
    \centering
    \ContinuedFloat*
    \includegraphics[width= \textwidth]{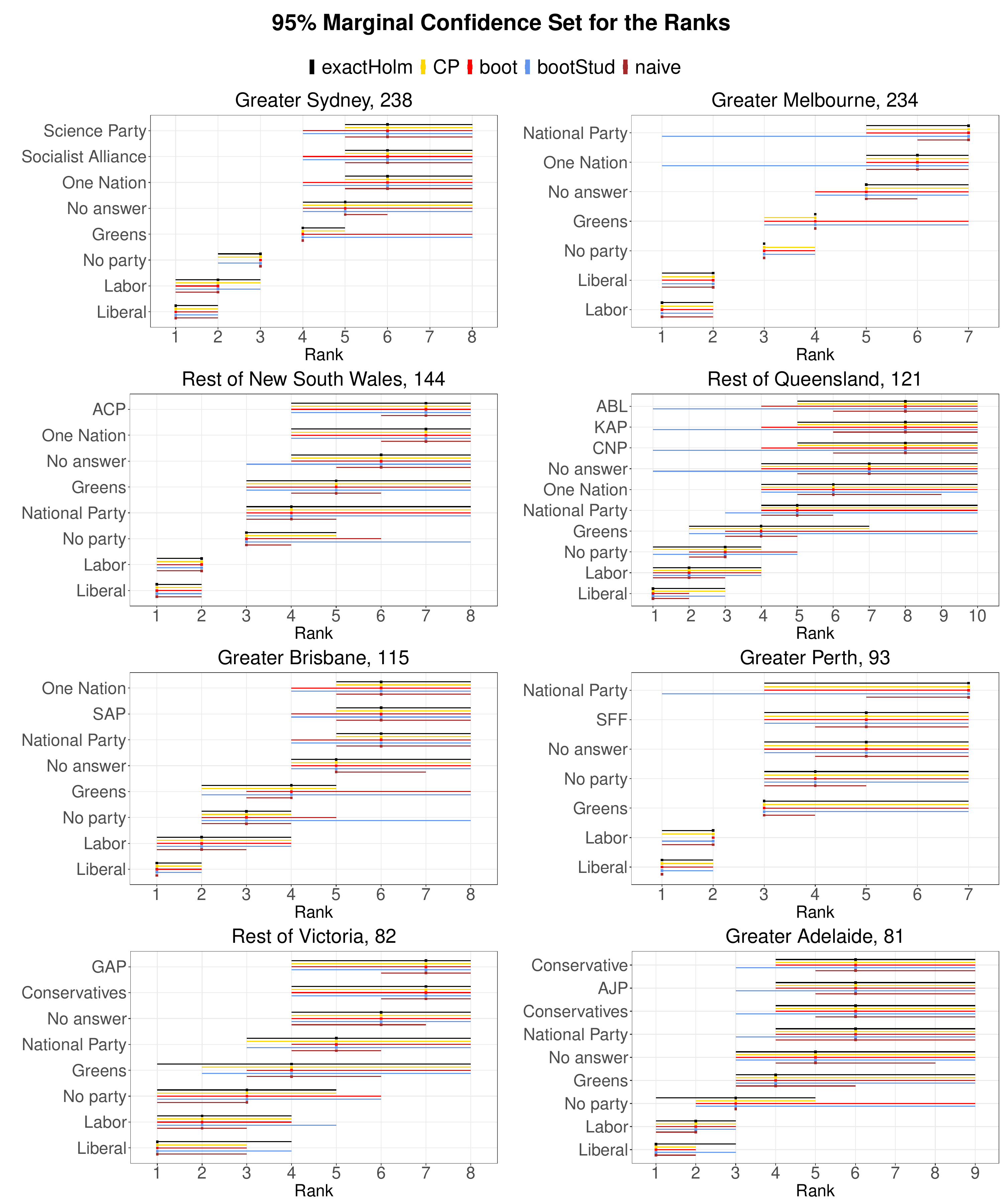}
    \caption{95\% marginal confidence sets for the ranks of categories in fifteen Australian territories ranked by their support share in AES2019. Each panel shows the confidence sets for the ranks computed by five methods for each party.}
    \label{fig: MargCS1}
\end{figure}

\begin{figure}[H]
    \centering
    \ContinuedFloat
    \includegraphics[width= \textwidth]{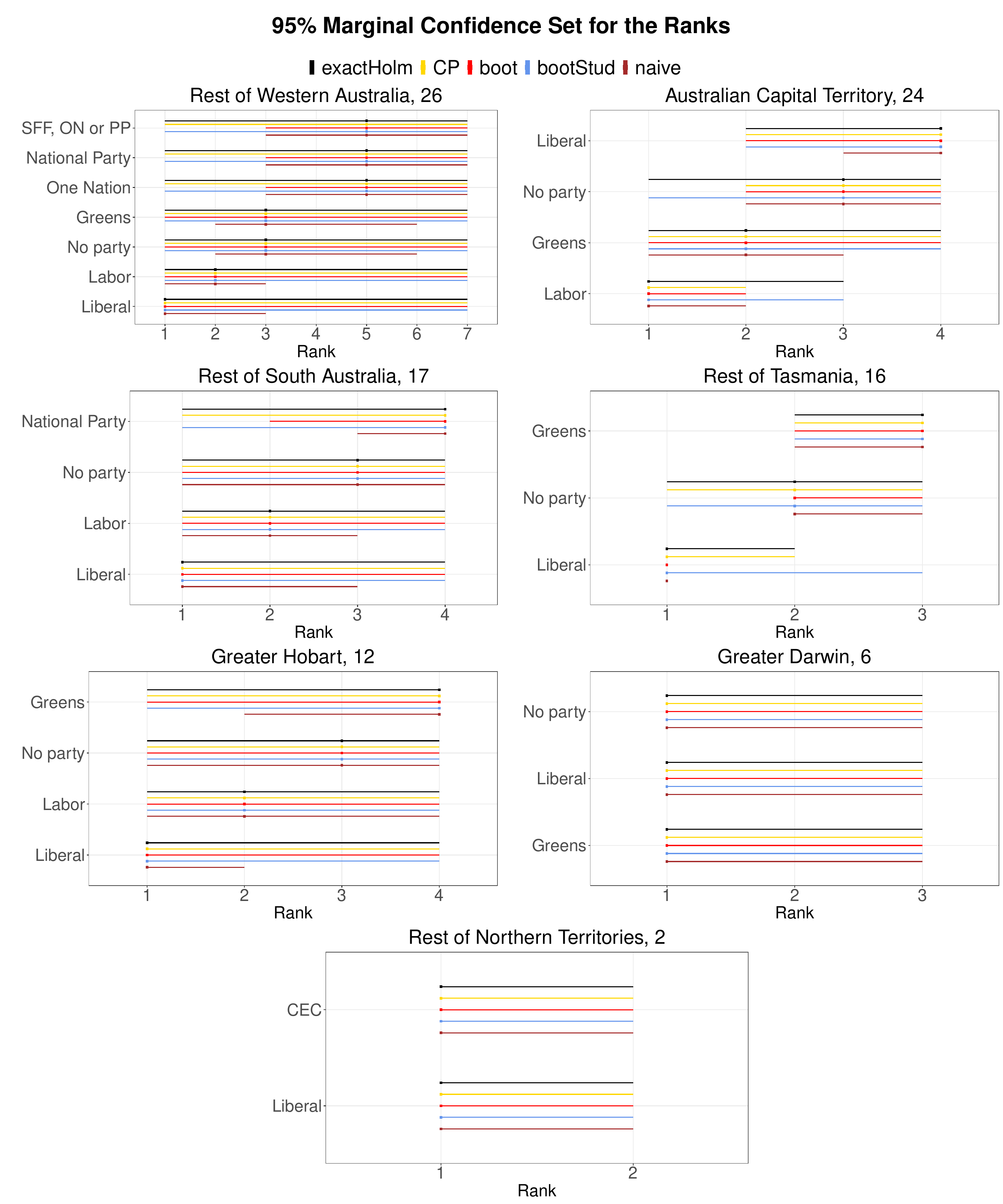}
    \caption{95\% marginal confidence sets for the ranks of categories in fifteen Australian territories ranked by their support share in AES2019. Each panel shows the confidence sets for the ranks computed by five methods for each party.}
    \label{fig: MargCS2}
\end{figure}

\begin{figure}[H]
    \centering
    \includegraphics[width=0.99\textwidth]{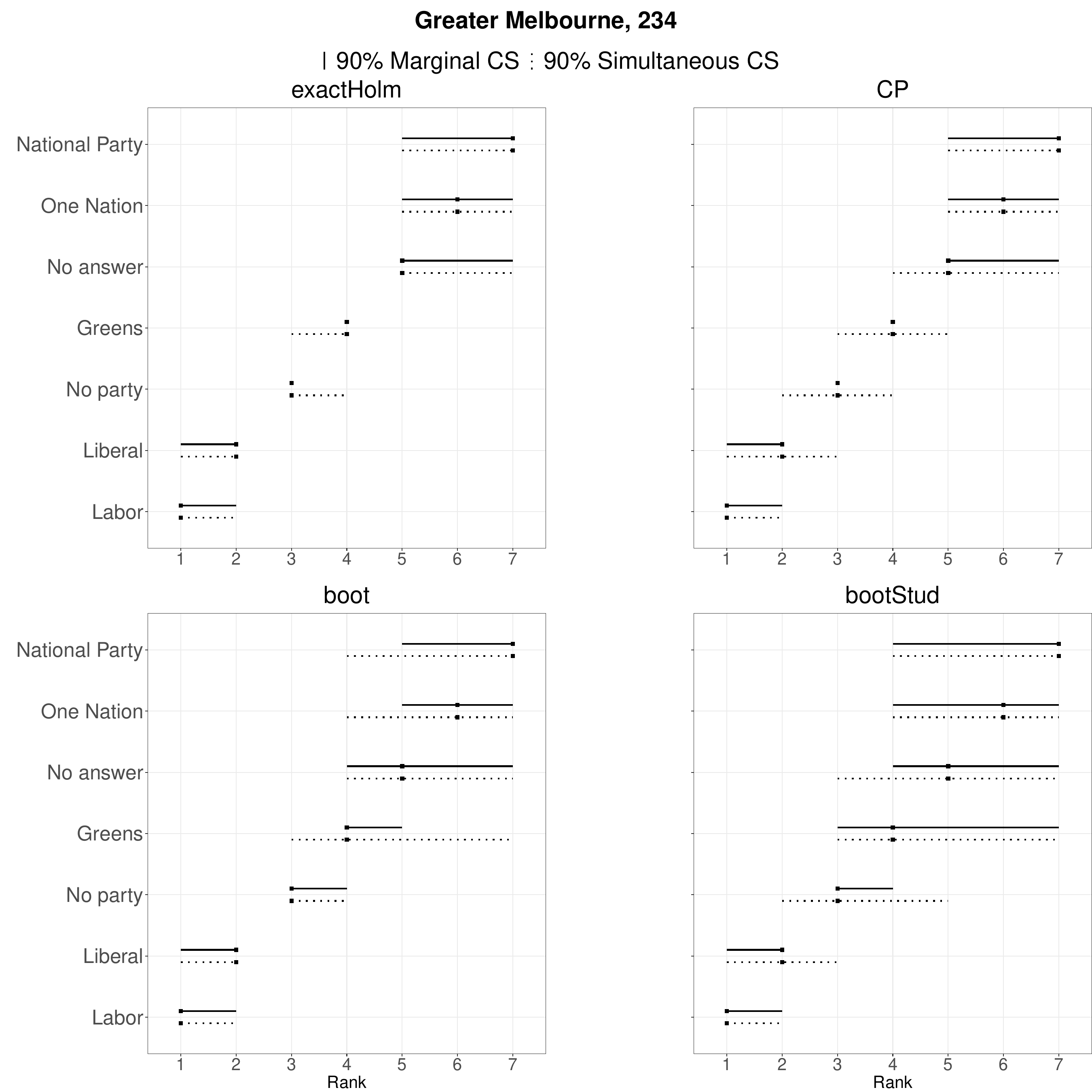}
    \caption{90\% marginal and 90\% simultaneous  confidence sets for the ranks of categories in Greater Melbourne. }
    \label{fig: MargSimul10}
\end{figure}

\clearpage
\newpage
\bibliographystyle{ims}
\bibliography{ref}

\begin{thebibliography}{17}
\expandafter\ifx\csname natexlab\endcsname\relax\def\natexlab#1{#1}\fi
\expandafter\ifx\csname url\endcsname\relax
  \def\url#1{\texttt{#1}}\fi
\expandafter\ifx\csname urlprefix\endcsname\relax\def\urlprefix{URL }\fi
\providecommand{\eprint}[2][]{\url{#2}}

\bibitem[{Andrews et~al.(2018)Andrews, Kitagawa and McCloskey}]{Andrews:2018qf}
\textsc{Andrews, I.}, \textsc{Kitagawa, T.} and \textsc{McCloskey, A.} (2018).
\newblock Inference on winners.
\newblock Working Paper CWP 31/18, CeMMAP.

\bibitem[{Bean et~al.(2019)Bean, Cameron, Gibson, Makkai, McAllister and
  Sheppard}]{beanAES}
\textsc{Bean, C.}, \textsc{Cameron, S.}, \textsc{Gibson, R.}, \textsc{Makkai,
  T.}, \textsc{McAllister, I.} and \textsc{Sheppard, J.} (2019).
\newblock Australian election study 2019: Voters technical report.
\newblock \textit{The Australian National University}.

\bibitem[{Brown et~al.(2001)Brown, Cai and DasGupta}]{Brown:2001p5787}
\textsc{Brown, L.~D.}, \textsc{Cai, T.~T.} and \textsc{DasGupta, A.} (2001).
\newblock Interval estimation for a binomial proportion.
\newblock \textit{Statistical Science}, \textbf{16} 101--117.

\bibitem[{Cameron and McAllister(2019)}]{cameronAES}
\textsc{Cameron, S.} and \textsc{McAllister, I.} (2019).
\newblock The 2019 australian federal election: Results from the australian
  election study.
\newblock \textit{The Australian National University}.

\bibitem[{Clopper and Pearson(1934)}]{Clopper:1934dd}
\textsc{Clopper, C.~J.} and \textsc{Pearson, E.~S.} (1934).
\newblock The use of confidence or fiducial limits illustrated in the case of
  the binomial.
\newblock \textit{Biometrika}, \textbf{26} 404--413.

\bibitem[{Goldstein and Spiegelhalter(1996)}]{Goldstein:1996re}
\textsc{Goldstein, H.} and \textsc{Spiegelhalter, D.~J.} (1996).
\newblock League tables and their limitations: Statistical issues in
  comparisons of institutional performance.
\newblock \textit{Journal of the Royal Statistical Society. Series A
  (Statistics in Society)}, \textbf{159} 385--443.

\bibitem[{Gu and Koenker(2020)}]{gu2020invidious}
\textsc{Gu, J.} and \textsc{Koenker, R.} (2020).
\newblock Invidious comparisons: Ranking and selection as compound decisions.
\newblock \textit{arXiv preprint arXiv:2012.12550}.

\bibitem[{Gupta and Panchapakesan(1979)}]{Gupta:1979hi}
\textsc{Gupta, S.~S.} and \textsc{Panchapakesan, S.} (1979).
\newblock \textit{Multiple Decision Procedures: Theory and Methodology of
  Selecting and Ranking Populations}.
\newblock John Wiley \& Sons, New York.

\bibitem[{Hall and Miller(2009)}]{Hall:2009oi}
\textsc{Hall, P.} and \textsc{Miller, H.} (2009).
\newblock Using the bootstrap to quantify the authority of an empirical
  ranking.
\newblock \textit{The Annals of Statistics}, \textbf{37} 3929--3959.

\bibitem[{Holm(1979)}]{Holm:1979xy}
\textsc{Holm, S.} (1979).
\newblock A simple sequentially rejective multiple test procedure.
\newblock \textit{Scandinavian Journal of Statistics}, \textbf{6} 65--70.

\bibitem[{Klein et~al.(2020)Klein, Wright and Wieczorek}]{Klein:2020oi}
\textsc{Klein, M.}, \textsc{Wright, T.} and \textsc{Wieczorek, J.} (2020).
\newblock A joint confidence region for an overall ranking of populations.
\newblock \textit{Journal of the Royal Statistical Society: Series C (Applied
  Statistics)}, \textbf{69} 589--606.

\bibitem[{Lehmann and Romano(2022)}]{Lehmann:2022aa}
\textsc{Lehmann, E.~L.} and \textsc{Romano, J.~P.} (2022).
\newblock \textit{Testing Statistical Hypotheses}.
\newblock 4th ed. Springer Cham.

\bibitem[{Marden(1995)}]{Marden:1995aa}
\textsc{Marden, J.~I.} (1995).
\newblock \textit{Analyzing and Modeling Rank Data}.
\newblock 1st ed. Chapman {\&} Hall/CRC.

\bibitem[{Mogstad et~al.(2024)Mogstad, Romano, Shaikh and
  Wilhelm}]{Mogstad:2024aa}
\textsc{Mogstad, M.}, \textsc{Romano, J.~P.}, \textsc{Shaikh, A.~M.} and
  \textsc{Wilhelm, D.} (2024).
\newblock Inference for ranks with applications to mobility across
  neighbourhoods and academic achievement across countries.
\newblock \textit{The Review of Economic Studies}, \textbf{91} 476--518.

\bibitem[{Romano and Shaikh(2012)}]{Romano:2012kj}
\textsc{Romano, J.~P.} and \textsc{Shaikh, A.~M.} (2012).
\newblock On the uniform asymptotic validity of subsampling and the bootstrap.
\newblock \textit{Annals of Statistics}, \textbf{40} 2798--2822.

\bibitem[{Romano and Wolf(2005)}]{romano2005gf}
\textsc{Romano, J.~P.} and \textsc{Wolf, M.} (2005).
\newblock Stepwise multiple testing as formalized data snooping.
\newblock \textit{Econometrica}, \textbf{73} 1237--1282.

\bibitem[{Xie et~al.(2009)Xie, Singh and Zhang}]{Xie:2009oi}
\textsc{Xie, M.}, \textsc{Singh, K.} and \textsc{Zhang, C.-H.} (2009).
\newblock Confidence intervals for population ranks in the presence of ties and
  near ties.
\newblock \textit{Journal of the American Statistical Association},
  \textbf{104} 775--788.

\end{thebibliography}

\end{document}